\documentclass[11pt]{article}
\usepackage[compact=,largepaper=true,paper=a4paper,font=Palatino,theorems=numbersfirst,figure=subfig,hypercolor=DarkSlateGrey]{ifiseries}
\usepackage{localKit}

\pgfdeclarelayer{background}
\pgfsetlayers{background,main}
\tikzset{%
vertex/.style={circle,fill=black!20,inner sep=0pt,minimum size=18pt},%
small vertex/.style={circle,fill=black!20,inner sep=0pt,minimum size=9pt},%
player/.style={circle,fill=black!20,inner sep=0pt,minimum size=18pt},%
marked player/.style={circle,draw=black,fill=white,inner sep=0pt,minimum size=17pt},%
dirlink/.style={draw,postaction={decorate,decoration=%
	{markings,mark=at position -.8pt with {\arrow[line width=1pt]{stealth'}}}}},%
subgraph/.style={draw,circle,decorate,decoration={coil,aspect=0},minimum size=2cm},%
every picture/.style={>=stealth',line width=.8pt},%
highlight/.style={draw,line width=5pt,-,black!30},%
on grid,%
auto,%
}

\newcommand{\vE}{\vec{E}}
\newcommand{\vG}{\vec{G}}
\newcommand{\diam}{\mathrm{diam}}
\newcommand{\Bi}{\mathsf{B}}
\newcommand{\Uni}{\mathsf{U}}
\newcommand{\nrm}[1]{|{#1}|}
\newcommand{\alignskip}{\quad\:}
\renewcommand{\switch}{\longrightarrow}

\newcommand{\PS}{\ensuremath{\mathsf{PS}}}

\newcommand{\nn}{\nu}
\newcommand{\sep}{\mathrm{sep}}
\newcommand{\payoff}{\pi}
\newcommand{\Rel}{\mathit{R}}
\newcommand{\tG}{\widetilde{G}}
\newcommand{\tV}{\widetilde{V}}
\newcommand{\tP}{\widetilde{P}}
\newcommand{\tC}{\widetilde{C}}
\newcommand{\shp}{\mathcal{P}}
\newcommand{\mmax}{m_{\max}}
\newcommand{\Emax}{E_{\max}}

\renewcommand{\S}{\ensuremath{\mathcal{S}}}
\newcommand{\Left}{\mathit{left}}
\newcommand{\Right}{\mathit{right}}

\newcommand{\bcc}{\textsc{bcc}}

\renewcommand{\P}{\mathrm{Pr}}
\renewcommand{\rel}{\mathrm{rel}}
\renewcommand{\SC}{C}
\renewcommand{\acro}[1]{#1}
\newcommand{\NE}{\ensuremath{\text{{\acro{NE}}}}}
\newcommand{\PNE}{{\ensuremath{\text{\acro{PNE}}}}}
\renewcommand{\PS}{{\ensuremath{\text{\acro{PS}}}}}
\newcommand{\MNE}{{\ensuremath{\text{\acro{MaxNE}}}}}
\newcommand{\ULF}{{\ensuremath{\text{\acro{ULF}}}}}
\newcommand{\BLF}{{\ensuremath{\text{\acro{BLF}}}}}
\renewcommand{\bcc}{{\ensuremath{\text{\acro{BCC}}}}}
\renewcommand{\OPT}{\mathrm{OPT}}
\renewcommand{\Uni}{\mathrm{U}}
\renewcommand{\Bi}{\mathrm{B}}
\newcommand{\subtree}[4]{%
	\begin{pgfonlayer}{background}
	\begin{scope}[shift=(#1),rotate=#2]
		\coordinate (subtree0) at (.6*#3,0);
		\coordinate (subtree1) at (0,0);
		\coordinate (subtree2) at (-30:#3);
		\coordinate (subtree3) at (30:#3);
		\path (subtree1) edge (subtree2)
			(subtree2) edge (subtree3)
			(subtree3) edge (subtree1);
		\node at (subtree0) {#4};
	\end{scope}
	\end{pgfonlayer}}
\newcommand{\TreeStructureSkel}{%
	\tikzset{node distance=2.5cm}
	\node[player] (c) { };
	\node[player] (u1) [right = 2.5 of c] {$u_1$};
	\node[player] (u2) [left = 2.5 of c] {$u_2$};
	\node[player] (u3) [below left = of c] {$u_3$};
	\node[player] (uN) [below right = of c] {$u_N$};
	\subtree{u1}{45}{2.5}{$T_1$}
	\subtree{u3}{225}{2}{$T_3$}
	\subtree{uN}{315}{2}{$T_N$}
	\path (c) edge node[swap] {$e_N$} (uN);
	\node (d1) [right = .8 of u3] { };
	\node (d2) [left = .8 of uN] { };
	\path (d1) edge[loosely dotted] (d2);}

\begin{document}

\gentitlepage%
{The Price of Anarchy for\\ Network Formation\\ in an Adversary Model}%
{}%
{Lasse Kliemann\\[1em]
\smaller Department of Computer Science\\
Christian-Albrechts-Universität zu Kiel\\
Christian-Albrechts-Platz 4\\
24118 Kiel, Germany\\
E-Mail:~\href{mailto:lki@informatik.uni-kiel.de}{lki@informatik.uni-kiel.de}\\
Web:~\url{http://lasse-kliemann.name}}%
{%
Permanent ID of this document: 92050128-8870-4eeb-9396-e969b41d7680.\\
This version is dated 2012-02-21.\\[1em]
\textbf{This is an extended and improved version of my article\\ in \textit{Games}, vol.~2, 2011~\cite{Kli11a}.}\\
Some constants have been improved in this version.
More importantly, 
it contains an additional section (\autoref{sec:bilateral}) on bilateral link formation.}

\newpage
\tableofcontents

\newpage
\begin{abstract}
\noindent
We study network formation
with $n$ \term{players}
and \term{link cost} $\alpha>0$.
After the network is built,
an adversary randomly deletes one link according to a certain probability distribution.
Cost for player $v$ incorporates 
the expected number of players to which $v$ 
will become disconnected.
We show existence of \term{equilibria}
and a \term{price of stability} of $1+o(1)$
under moderate assumptions on the adversary and $n\geq 9$.
\par
As the main result, we prove bounds on the \term{price of anarchy}
for two special adversaries:
one removes a link chosen uniformly at random,
while the other removes a link 
that causes a maximum number of player pairs to be separated.
For \term{unilateral} link formation
we show a bound of $O(1)$ on the price of anarchy for both adversaries,
the constant being bounded by $10+o(1)$ and $8+o(1)$, respectively.
For \term{bilateral} link formation 
we show $O(1+\sqrt{\sfrac{n}{\alpha}})$ for one adversary (if $\al>\frac{1}{2}$),
and $\Theta(n)$ for the other
(if $\al>2$ considered constant and $n \geq 9$).
The latter is the worst that can happen
for any adversary in this model (if $\al=\Om(1)$).
This points out substantial differences between unilateral
and bilateral link formation.
\end{abstract}

\section{Network Formation}
In network formation,
a multitude of individuals, called \term{players},
form a network in such a way that each player decides for herself 
to which other players she would like to connect.
So players can be considered vertices\footnote{We use \enquote{player} and \enquote{vertex} synonymously.}
of a (to-be-built) network.
Any outcome of this, \ie any network,
can be evaluated from the point of view of each player
via an \term{individual cost}.
Individual cost comprises \term{building cost},
proportional to the number of links\footnote{%
We use \enquote{link} and \enquote{edge} synonymously.}
built by the player,
and \term{indirect cost},
which expresses properties of the network.
\term{Social cost} is the sum of individual cost over all players.
There are two parameters: $n$ is the number of players
and $\alpha>0$ is the cost of a link.
Another crucial feature is how links can be formed:
unilaterally or bilaterally.
Under \term{unilateral link formation}
a player can connect to any other player
and is charged the amount of $\alpha$ for each link.
Under \term{bilateral link formation},
the consent of both endpoints is required
and if they both agree, then they pay $\alpha$ each.
When the network is so that 
no player sees a way to improve her individual cost,
we speak of an \term{equilibrium}.
The finer facets of the equilibrium concept 
have to be chosen according to the link formation rule:
\term{Nash equilibrium} is well suited for unilateral link formation,
whereas for bilateral link formation,
\term{pairwise Nash equilibrium} or \term{pairwise stability} are better suited;
definitions will be given later.
When the social cost is minimal, we speak of an \term{optimum}.\footnote{%
Optimal networks are also called \enquote{efficient} in the literature.}
The \term{price of anarchy}~\cite{KP99,Pap01} measures overall performance loss
due to distributed operation,
compared to when a central authority would enforce an optimum:
the price of anarchy is the \emphasis{worst-case} ratio of the social cost of an equilibrium
to that of an optimum.
One is interested in bounds on the price of anarchy,
especially upper bounds.
The price of anarchy is a static measure
in the sense that it does not try to assess
how a network might evolve over time.
It instead builds upon the assumption that
equilibrium networks will emerge
from evolutionary processes.
A related concept is \term{price of stability},
\ie the \emphasis{best-case} ratio of the social cost of an equilibrium
to that of an optimum.
This work's focus is on upper bounds on the price of anarchy,
although we prove tight bounds on the price of stability 
and some structural results along the~way.
\paragraph{Our Contribution.}
We study the price of anarchy in an \emphasis{adversary model}.
After the network is built, an adversary deletes
exactly one link from it.
The adversary is modeled by a random experiment;
hence in general there is an uncertainty which link will be deleted,
but players know the probability distribution according 
to which the adversary chooses the link to destroy.
Indirect cost of a player $v$ is defined as the expected number of players 
to which $v$ will lose connection when the adversary strikes.
Formally, an adversary is a mapping from networks
to probability distributions on the edges of the particular network.
\par
Although it appears limiting that the adversary can only destroy one link,
this model already is challenging to analyze.
It is a contribution to the understanding of how networks are formed
when it is important that every vertex can reach every other vertex,
for example for data transmission or delivery of~goods.
\par
After preparations and discussion of related work 
(\autoref{sec:model-framework} to \autoref{sec:bridge-tree})
we start out with a simple $O(1+\frac{n}{\alpha})$ bound 
on the price of anarchy for any adversary 
and independent of the link formation rule and the equilibrium concept,
but under the assumption that equilibria only have a linear (in $n$) number of edges
(\autoref{sec:simple-bound}).
This assumption will later be shown to be valid
for the two special adversaries under consideration
and unilateral link formation.
In the three sections that follow (\autoref{sec:uni-opt-ne-stability} 
to \autoref{sec:uni-smart}), we consider unilateral link formation.
We constructively show existence of Nash equilibria
under some moderate assumptions on the adversary and $n\geq 9$,
including a co-existence of two very different topologies
for the same range of parameters.
A~$1+o(1)$ bound on the price of stability follows from our existence results.
Then for two specific adversaries
we improve the $O(1+\frac{n}{\al})$ bound on the price of anarchy to $O(1)$.
These two adversaries are chosen to represent extreme cases:
the first one, called \term{simple-minded},
chooses one link uniformly at random.
The second one, called \term{smart},
chooses uniformly at random 
from the set of those links whose removal causes 
a maximum number of vertex pairs to be separated.
The proof techniques for the simple-minded adversary
are roughly similar to what has been used
for other models before, \eg~\cite{FLM+03},
namely we relate to the diameter of equilibria.
For the smart adversary, a new approach has to be taken;
it works by an appropriate decomposition of the graph.
\par
Finally, we consider bilateral link formation (\autoref{sec:bilateral}).
The constructions done in \autoref{sec:uni-opt-ne-stability} 
for unilateral link formation
carry over with little effort,
and so does the bound on the price of stability.
For the simple-minded adversary,
we show an $O(1+\sqrt{\sfrac{n}{\alpha}})$ bound on the price of anarchy if $\al > \frac{1}{2}$
and $O(1+\sqrt{\sfrac{n}{\alpha^{1.5}}})$ otherwise.
For the smart adversary,
we show an $O(1+\frac{n}{\alpha})$ upper bound
and an $\Omega(1+\frac{n}{\alpha})$ lower bound,
the latter requiring $\alpha>2$ and $n \geq 9$.
So if $\alpha$ is considered a constant and $\alpha>2$ (and $n \geq 9$),
the price of anarchy for the smart adversary
jumps from $O(1)$ to the worst possible, namely $\Omega(n)$,
when switching from unilateral to bilateral link formation.
\par
We also consider convexity of cost.
This is important for the relation between 
the two equilibrium concepts used for
bilateral link formation.
If cost is convex, 
then pairwise Nash equilibrium and pairwise stability are equivalent.
We prove convexity for the simple-minded adversary.
For the smart adversary, we disprove convexity;
yet we have to leave open the adjacent question
whether pairwise Nash equilibrium and pairwise stability
in fact diverge in this case.
\paragraph{Open Problems.}
Tight bounds on the price of anarchy for other adversaries,
or for a general adversary are left for future work.
As one of the most intriguing open problems,
we leave the case of an adversary removing more than one link.
Since our proofs rely heavily on the restriction of only one link being removed,
this is expected to be a new challenge.
Finally, the jump from $O(1)$ to $\Omega(n)$ 
induced by a switch from unilateral to bilateral link formation
raises the question how our model behaves with other link formation rules
\eg in the recently introduced models with transfers~\cite{BJ07}.

\section{Model Framework}
\label{sec:model-framework}%
We give a rigorous description of the model framework
that will be used in the following.
Let $n\geq 3$
and $V$ a set of $n$ vertices, say $V=\setn{n}\df\set{1,\hdots,n}$.
Each vertex represents an individual, called a \term{player}.
Each player names a list of other players to which she would like to build an edge.
The decisions of player $v$ are collected in a vector $S_v\in\OI^n$,
with $S_{vw}=1$ meaning that $v$ would like to have the edge $\set{v,w}$
in the network.
Such an $S_v$ is called a \term{strategy} for player~$v$.
A vector of strategies $S=(\eli{S}{n})$, one for each player,
is called a \term{strategy profile}.
A~strategy profile can be written as a matrix $\OI^{n\times n}$ 
and interpreted as the adjacency matrix of a directed graph $\vG(S)=(V,\vE(S))$.
Then $(v,w)\in \vE(S)$ if and only if 
player $v$ would like to have the edge $\set{v,w}$.
We will often work with this representation of strategy profiles.
Denote $\S(n)$ the set of all strategy profiles for $n$ players.
We use sets $F\subseteq V\times V$ to denote strategy changes.
Define $S+F$ and $S-F$ by
setting for all $x,y\in V$
\begin{equation*}
(S+F)_{xy} \df \begin{cases}
1 & \text{if $(x,y)\in F$}\\
S_{xy} & \text{otherwise}
\end{cases}\quad\text{and}\quad
(S-F)_{xy} \df \begin{cases}
0 & \text{if $(x,y)\in F$}\\
S_{xy} & \text{otherwise}
\end{cases}\enspace.
\end{equation*}
If $F=\set{(v,w)}$, we write $S+(v,w)$ and $S-(v,w)$.
For instance, $S+(v,w)$ means that we add to $S$
the request of player $v$ for the edge $\set{v,w}$.
\par
The graph which is actually built is called the \term{final graph} and denoted $G(S)$.
We are interested in two different versions of the final graph:
\begin{itemize}
\item The \term{unilateral final graph},
denoted $G^\Uni(S)\df(V,E^\Uni(S))$, where
\begin{equation*}
E^\Uni(S)\df\setst{\set{v,w}}{S_{vw}=1\lor S_{wv}=1}\enspace.
\end{equation*}
So the wish of one endpoint, either $v$ or $w$, 
is enough to have $\set{v,w}$ in the final graph.
We also call this \term{unilateral link formation} (\ULF).
Throughout \autoref{sec:uni-opt-ne-stability} to \autoref{sec:uni-smart}
we will only consider \ULF.
\item The \term{bilateral final graph},
denoted $G^\Bi(S)\df(V,E^\Bi(S))$, where
\begin{equation*}
E^\Bi(S)\df\setst{\set{v,w}}{S_{vw}=1\land S_{wv}=1}\enspace.
\end{equation*}
So both endpoints, $v$ and $w$, have to agree on having $\set{v,w}$ in the final graph.
Otherwise, it will not be built.
We also call this \term{bilateral link formation} (\BLF).
We will consider \BLF\ in \autoref{sec:bilateral}.
\end{itemize}
We omit the \enquote{$\Uni$} and \enquote{$\Bi$} superscripts
when a statement or definition addresses both versions,
or when a restriction to a particular version is clear from context.
\par
We speak of \term{selling} (or \term{deleting}, \term{removing}) an edge $e$ 
if a player changes her strategy so that $e$ is no longer
part of the final graph.
We speak of \term{buying} (or \term{adding}, \term{building}) an edge $e$ 
if a player changes her strategy so that $e$ is then 
part of the final graph.
\par
Fix parameters $n$ and $\alpha>0$.
Given a strategy profile $S\in\S(n)$ 
each player experiences a cost $C_v(S)$, her \term{individual cost}.
It is comprised of \term{building cost} plus \term{indirect cost}.
Building cost is computed by counting $\alpha$ for each edge
that $v$ requested.\footnote{%
This deviates from the definition in~\cite{Kli10a}.
However, for \term{essential} strategy profiles (defined below),
the definitions given here and in~\cite{Kli10a} coincide.}
Indirect cost can be defined in many different ways 
and usually captures properties of the final graph,
we denote it $I_v(G(S))$
and sometimes just $I_v(S)$ for a streamlined notation.
Denoting $\nrm{S_v}\df\sum_{w\in V} S_{vw}$, 
we can write out the individual cost
$C_v(S) \df \nrm{S_v} \, \alpha + I_v(G(S))$.
An equivalent concept found in the literature is \term{payoff}:
properties of the final graph are expressed by \term{income},
and payoff is income minus building~cost.
\par
The indirect cost $I_v(\cdot)$ is a placeholder 
to be filled in in order to have a concrete model.
For example, the model in~\cite{FLM+03} uses $I_v(G)=\sum_{w\in V}\dist_G(v,w)$,
where the \term{distance} $\dist(v,w)$ is the length of a shortest path
between $v$ and $w$ and equals $\infty$ if there is no such path.
We call this the \term{sum-distance model}.
The price of anarchy in the sum-distance model is particularly well-studied.
We will introduce our definition of indirect cost
in \autoref{sec:adversary-model}.
\par
We call indirect cost \term{anonymous} if
for each final graph $G=(V,E)$
and each graph automorphism $\phi:V\map V$ of $G$,
we have $I_v(G)=I_{\phi(v)}(G)$ for all $v\in V$.
In other words, anonymity of indirect cost means that
$I_v(G)$ does not depend on $v$'s identity,
but only on $v$'s position in the final graph $G$.
This is of importance in particular if $G$ has symmetry.
For instance, if $G$ is a cycle, 
then all vertices experience the same indirect cost.
If $G$ is a path, then both endpoints experience the same indirect cost.
\par
The \term{social cost} of $S$ is $\SC(S)\df\sum_{v\in V}C_v(S)$.
When we sum up the building cost over all players, we also speak of \term{total building cost};
when we sum up the indirect cost over all players, we speak of \term{total indirect cost}.
Hence social cost is total building cost plus total indirect cost.
A~strategy profile $S^*$ is called \term{optimal} if 
$\SC(S^*)=\min_{S\in\S(n)} \SC(S)$.
This is with respect to fixed $\alpha$;
denote $\OPT(n,\alpha)$ the social cost of an optimum
for given $n$ and~$\alpha$.
An undirected graph $G$ is called optimal
if $G=G(S^*)$ for an optimal~$S^*$.
\par
A strategy profile $S$ is called \term{essential}\footnote{%
The \enquote{essential} term was used in~\cite{BG00a} in the context of ULF.
In~\cite{CG07}, the concept is called \enquote{non-superfluous} in the context of BLF.
In~\cite{Kli10a}, \enquote{clean} was used the way we use \enquote{essential} here.}
if for all \mbox{$v,w\in V$}
the following implication holds:
$S_{vw}=1 \Longrightarrow S_{wv}=0$ when using \ULF,
and $S_{vw}=1 \Longrightarrow S_{wv}=1$ when using \BLF.
In other words,
an essential strategy profile does not contain \term{unnecessary} or \term{useless} requests,
respectively.
In an essential strategy profile,
building cost deserves its name in the following sense:
players pay only for edges that would not be in the final graph
without this payment (relevant for \ULF)
and which actually appear in the final graph (relevant for \BLF).
This means that in \ULF,
each edge in the final graph is paid for $\alpha$ by exactly one of its endpoints,
namely the one who requested it.
In \BLF,
each edge in the final graph is paid for $\alpha$ by each of its endpoints.
Hence if $S$ is essential then
with \ULF\ $\SC(S) = \card{E(S)} \, \alpha + \sum_{v\in V} I_v(G(S))$,
and with \BLF\ $\SC(S) = 2 \card{E(S)} \, \alpha + \sum_{v\in V} I_v(G(S))$.
Social cost then only depends on the final graph.
In \BLF, an essential $S$ is fully determined by its final graph $G(S)$.
\par
For each strategy profile $S$, 
dropping all unnecessary or useless requests
results in an essential strategy profile $S'$ with the same final graph
and with the same or a smaller individual cost for each player.
Moreover, it is easy to see 
that if $S$ is an equilibrium of any of the three kinds introduced below,
then $S'$ is an equilibrium of the same kind.
It is hence reasonable to restrict to essential strategy profiles,
and we will do so in the following.
\par
\ULF\ and \BLF\ are \emphasis{similar} regarding link removal:
in both,
if $v$ pays for a link~$e$, then $v$ can have the link removed unilaterally,
\ie by changing her strategy whilst the strategies of all other players are maintained.
(For \ULF, we rely on the restriction to essential strategy profiles here.)
\ULF\ and \BLF\ are \emphasis{different} regarding link formation:
there is no way for $v$ to form a link unilaterally when \BLF\ is in effect.
(We rely on the restriction to essential strategy profiles again here.)
However, with \ULF, a player can form any link unilaterally.
\par
Recall that we can specify strategy profiles as directed graphs.
Furthermore,
since the social cost is fully determined by the final graph
(since we restrict to essential strategy profiles),
it suffices to consider the final graph 
(which is an undirected graph) in places
where only the social cost is relevant.
It is also sufficient to consider the final graph
when dealing with \BLF,
since the final graph under \BLF\ fully specifies 
the underlying strategy profile
(restricting to essential ones).
\par
A strategy profile $S$ is called a \term{Nash equilibrium} (\NE)
if no player can strictly improve her individual cost 
by changing her strategy
given the strategies of the other players, \ie
\begin{equation*}
C_v(S+A-D) \geq C_v(S) \quad\quad \forall A,D\subseteq\set{v}\times V \quad \forall v\in V\enspace.
\end{equation*}
Denote the set of all \NE\ for given $n$ and $\alpha$ by $\NE(n,\alpha)$.
An undirected graph $G$ is called a \NE\
if there exists a \NE\ $S$ such that $G=G(S)$.
The \term{price of anarchy} (with respect to \NE)
is the social cost 
of a worst-case \NE\ divided by the social cost of an optimum,
\ie
\begin{equation*}
\frac{\max_{S\in\NE(n,\alpha)} \SC(S)}{\OPT(n,\alpha)} 
\enspace.
\end{equation*}
When we replace \enquote{$\max$} for \enquote{$\min$},
we speak of \term{price of stability}.
Both notions are meant relative to a given $n$ and $\al$.
They extend naturally to other equilibrium concepts, instead of \NE.
\par
We use \NE\ as equilibrium concept for \ULF.
However, \NE\ is \emphasis{not} an adequate equilibrium concept for \BLF.
With \BLF, the empty graph is a \NE\
regardless of other properties of the model,
since link formation cannot happen unilaterally.
This not only appears unreasonable,
it also trivially pushes the price of anarchy to $\infty$ for many models,
including the sum-distance model.
A~remedy  is to introduce a minimum of cooperation:
absence of a link from the final graph
requires the additional justification that adding this link
would be an impairment to at least one of its endpoints.
Otherwise it shall be built.
This is formalized in the following definition.
A~strategy profile $S$ is called a \term{pairwise Nash equilibrium} (\PNE)
if it is a \NE\ and for all $v,w\in V$
such that $\set{v,w}\not\in E(S)$ the following implication holds:\footnote{%
This can also be found in the literature with strict inequality in the premise.
Both variants have their advantages and disadvantages.
The advantage of ours is that it rules out the empty strategy profile 
as \PNE\ whenever indirect cost $\infty$ is assigned to disconnected final graphs.
A disadvantage is that in order to relate \NE\ to \PNE,
as we will do in \autoref{sec:bilateral},
a refinement of \NE\ is required, 
which we call \enquote{maximal Nash equilibrium} (\MNE).
It will be introduced later.}%
\begin{equation}
\label{eqn:pne}%
C_v(S+(v,w)+(w,v)) \leq C_v(S)
\:\Longrightarrow\:
C_w(S+(v,w)+(w,v)) > C_w(S) 
\enspace.
\end{equation}
We call an undirected graph $G$ a \PNE\ if there
exists a strategy profile $S$ being a \PNE\ and $G=G(S)$.
\par
A related concept is \term{pairwise stability} (\PS).
A strategy profile $S$ is called \term{pairwise stable} (\PS)
if condition~\eqref{eqn:pne} holds for all $\set{v,w}\not\in E(S)$
and if for all $\set{v,w}\in E(S)$:\footnote{%
In fact, one of the two conditions, say $C_v(S-(v,w))\geq C_v(S)$, 
clearly suffices.}
\begin{equation*}
C_v(S-(v,w)) \geq C_v(S)
\quad \text{and} \quad C_w(S-(w,v)) \geq C_w(S)
\enspace.
\end{equation*}
We call an undirected graph $G$ \PS\ if there
exists a strategy profile $S$ being \PS\ and $G=G(S)$.
So, with \PS, only single-link deviations have to be considered.
\par
We make the convention that whenever we speak of \NE\ (or \MNE, introduced below),
we mean that relative to \ULF.
Whenever we speak of \PNE\ or \PS,
we mean that relative to \BLF.
\par
Clearly, a \PNE\ is also \PS,
and so for fixed $n$ and $\al$,
the price of anarchy with respect to \PNE\ is upper-bounded
by the price of anarchy with respect to \PS.
The converse holds if cost is convex
on the set of \PS\ strategy profiles.
Convexity of cost relates removal of multiple links
to removal of each of those links alone.
This addresses the difference between \PNE\ and \PS:
in the former, removal of multiple links has to be considered,
whereas the latter is only concerned with removal of single links.
Let $v\in V$ and $S$ a strategy profile.
We call $C_v$ \term{convex in $S$}
if for all $\set{w_1,\hdots,w_k}\subseteq V$ we have
\begin{equation*}
C_v\big(S-(v,w_1)-\hdots-(v,w_k)\big) - C_v(S)
\geq \sum_{i=1}^k \big(C_v(S-(v,w_i)) - C_v(S)\big)\enspace,
\end{equation*}
or, equivalently,
\begin{equation}
\label{eqn:def-convex-2}%
I_v\big(S-(v,w_1)-\hdots-(v,w_k)\big) - I_v(S)
\geq \sum_{i=1}^k \big(I_v(S-(v,w_i)) - I_v(S)\big)\enspace.
\end{equation}
We call $C_v$ \term{convex} on a set of strategy profiles $\mathcal{S}$,
if it is convex in every $S\in\mathcal{S}$.
We call $C_v$ \term{convex} if it is convex on $\S(n)$, 
\ie the set of all strategy profiles for the given number $n$ of players.
We say that \enquote{cost is convex} if $C_v$ is convex for each player $v$.
It was shown by Corbo and Parkes~\cite{CP05},
their proof being based on a result by Calv\'{o}-Armengol 
and \.{I}lkili\c{c}~\cite{CI05},
that cost is convex in the sum-distance model.
Hence, in the sum-distance model \PNE\ and \PS\ coincide.
\par
For both, \PNE\ and \PS, condition~\eqref{eqn:pne} implies
that if indirect cost $\infty$ is assigned to a disconnected final graph,
all \PNE\ and \PS\ graphs are connected.
(The same holds for \NE.)
\par
To study the relation between \NE\ (with \ULF)
and \PNE\ (with \BLF)
we need a refinement of \NE,
not widely known in the literature.
We call a \NE\ $S$ a \term{maximal Nash equilibrium} (\MNE)
if $C_v(S+(v,w_1)+\hdots+(v,w_k)) > C_v(S)$
for all $\set{v,w_1},\hdots,\set{v,w_k}\not\in E(S)$.
That is, we exclude the possibility 
that a player can buy additional links so
that the gain in her indirect cost and the additional building cost
nullify each other.
We will require this notion in \autoref{sec:bilateral}.
\begin{remark}
\label{rem:mne}%
A \NE\ $S$ is maximal
if indirect cost $I_v(S)$ has its minimum possible value for all players $v$
(which is $0$ for most models).
A \NE\ is also maximal,
if there exists $\eps>0$
such that it is still a \NE\ for link cost $\alpha-\eps$ instead of $\alpha$.
Hence, if $S$ is a \NE\ for all $\alpha\geq f(n)$,
for some function~$f$,
this implies that $S$ is a \MNE\ for all~$\alpha>f(n)$.
\end{remark}
\par\pagebreak
We require some basic graph-theoretic notions.
Let an undirected graph $G=(V,E)$ be given,
that is, $V$ is a finite set and $E\subseteq {V\choose 2}$.
A~\term{walk} of length $\ell$ is a sequence of vertices $W=(\elix{v}{0}{\ell})$
such that $\set{v_{i-1},v_i}\in E$ for all $i\in\setn{\ell}$.
Denote 
$V(W)\df\set{\elix{v}{0}{\ell}}$ its vertices
and $E(W)\df\setst{\set{v_{i-1},v_i}}{i\in\setn{\ell}}$ its edges.
The walk is called a \term{path} if all its vertices are distinct,
that is, if $\card{V(W)}=\ell+1$.
The walk is called a \term{cycle} if
it has at least length $3$ (\ie $\ell \geq 3$) and
all its vertices except the last are distinct
(\ie $\card{\set{\elix{v}{0}{\ell-1}}}=\ell$)
and the walk is closed (\ie $v_0=v_\ell$).
Sometimes we use a notation
that gives names to the edges in the walk,
like $(v_0,e_1,v_1,e_2,\hdots,e_{\ell},v_{\ell})$.
If $C$ is a cycle
and $e=\set{u,w}$ is an edge with $u,w\in V(C)$
but $e\not\in E(C)$,
we call $e$ a \term{chord}.
For a subset $W\subseteq V$ denote $G[W]\df(W,\,{W \choose 2} \cap E)$
the \term{induced subgraph} of $W$,
or the \term{subgraph induced} by $W$.
If $G$ is a graph, then $V(G)$ denotes its set of vertices
and $E(G)$ its set of edges;
this is useful when $G$ was not introduced writing \enquote{$G=(V,E)$}.
More graph-theoretic notions will be introduced
along the way as we need them.
\par
One might suggest using multigraphs instead of graphs,
since in our adversary model, connectivity under removal of edges is relevant.
However, none of our results becomes false when we allow multigraphs.
Where not obvious, a remark on this is made.
So we can stick to the simpler notion of graphs.
\par
In order to not have to introduce names for all occurring constants,
we use \enquote{$O$} and \enquote{$\Omega$} notation.
For our results, we use this notation in the following understanding 
(it does not necessarily apply to all cited results).
We write
\enquote{$x=O(y)$}
if there exists a constant $c>0$ such that
$x \leq c y$.
The constant may only depend
on other constants and is in particular
independent of the non-constant quantities that constitute $x$ and $y$,
\eg parameters $n$ and $\al$.
We do not implicitly require that some quantities, \eg $n$,
have to be large.
Analogously, we write
\enquote{$x=\Omega(y)$}
if there exists a constant $c>0$ such that
$x \geq c y$.
Note that
\enquote{$O$} indicates an upper bound, 
making no statement about a lower bound;
while \enquote{$\Omega$} indicates a lower bound,
making no statement about an upper bound.
We write $x=\Theta(y)$ if $x=O(y)$ and $x=\Omega(y)$;
the constants used in the \enquote{$O$}
and the \enquote{$\Omega$} statement may be different, of course.
\par
The \enquote{$o$} notation is only used in one form,
namely $o(1)$ substituting a quantity that tends to $0$
when $n$ tends to infinity,
regardless whether other parameters are fixed or not.
Whenever we write \enquote{$o(1)$} in an expression,
it is meant as an upper bound, making no statement about a lower bound.

\section{Adversary Model}
\label{sec:adversary-model}%
An \term{adversary} $\mathcal{A}$ is a mapping
assigning to each graph $G=(V,E)$
a probability measure $\P_G^{\mathcal{A}}$ on the edges $E$ of $G$.
Given a connected graph $G$,
the \term{relevance} of an edge $e$ for a player $v$ is the number 
of vertices that can, starting at $v$, \emphasis{only} be reached via~$e$.
We denote the relevance by $\rel_G(e,v)$
and the sum of all relevances for a player 
by $\Rel_{G}(v)\df\sum_{e\in E} \rel_{G}(e,v)$.
An edge of a connected graph is called a \term{bridge}
if its removal destroys connectivity,
or equivalently, if it is no part of any cycle.
The relevance $\rel_G(e,v)$ is $0$ iff $e$ is \emphasis{not} a bridge.
Given a strategy profile $S$ where $G(S)$ is connected,
we define the \term{individual cost} of a player $v$ by
\begin{equation*}
C_v(S) \df \nrm{S_v} \, \alpha 
	+ \sum_{e\in E(S)} \rel_{G(S)}(e,v) \,\, \Pr[_{G(S)}^{\mathcal{A}}]{\set{e}}\enspace.
\end{equation*}
The indirect cost is the expected number of vertices
to which $v$ will lose connection when exactly one edge is removed 
from $G(S)$ randomly and according to the probability measure given by 
the adversary~$\mathcal{A}$.
For this indirect cost,
we use the term \term{disconnection cost} in the following
instead of \enquote{indirect cost}.
We define the indirect cost to be $\infty$ when $G(S)$ is not connected,
so we can concentrate on connected graphs in our study of
optima and equilibria.
We usually omit the \enquote{$G(S)$} subscripts
and also the \enquote{$\mathcal{A}$} superscript;
we also write \enquote{$E$} instead of \enquote{$E(S)$}
and \enquote{$m$} for the number of edges, \ie $m=\card{E(S)}$.
\begin{remark*}
Since $\infty$ is assigned to disconnected final graphs,
optima, NE, and PS graphs are connected.
\end{remark*}
\begin{proof}
This is clear for optima and also for NE (under ULF):
since a connected graph has finite indirect cost,
a player would always choose to build enough links 
in order to make the graph connected.
For PS (under BLF) it is a consequence of having non-strict inequality
in the premise of~\eqref{eqn:pne}.
\end{proof}
\par
The \term{separation} of an edge $e$, denoted $\sep(e)$,
is the number of ordered vertex pairs that will be separated by the removal of $e$.
For a bridge $e$, 
denote $\nn(e)$ the number of vertices in the component of $G-e$
that has a minimum number of vertices;
we have $\nn(e)\leq\sfloor{\frac{n}{2}}$.
If~$e$ is not a bridge, we define $\nn(e)\df0$.
Then $\sep(e) = 2 \nn(e) \, (n-\nn(e))$
and also $\sep(e) = \sum_{v\in V} \rel(e,v)$.
If $e$ is a bridge, then $\sep(e) \geq 2 \, (n-1)$.
We can express the social cost now
(total building cost given for \ULF):
\begin{equation*}
\SC(S) \df \sum_{v\in V} C_v(S)
= m \, \alpha + \sum_{v\in V} 
\sum_{e\in E} \rel(e,v) \, \Pr{\set{e}} 
= m \, \alpha + \sum_{e\in E} \sep(e) \, \Pr{\set{e}}
\enspace.
\end{equation*}
\par
We call an adversary \term{symmetric} if 
for a fixed graph, the probability of an edge only
depends on its separation, \ie
$\sep(e) = \sep(e')$ implies $\Pr{\set{e}) = \P(\set{e'}}$ for all $e,e'$.
The following proposition is proved straightforwardly.
\begin{proposition}
\label{prop:symmetric-is-anonymous}%
A symmetric adversary induces anonymous disconnection cost.
\end{proposition}
\begin{proof}
Let $G=(V,E)$ be connected
and $\phi:V\map V$ a graph automorphism of $G$.
If $e=\set{v,w}$ is a non-bridge, then $\phi(e)\df\set{\phi(v),\phi(w)}$
is a non-bridge as well, and so $\sep(e)=0=\sep(\phi(e))$
and $\rel(e,v)=0=\rel(\phi(e),\phi(v))$ for all $v\in V$.
Let $e=\set{v,w}$ be a bridge
and $G_1$, $G_2$ be the two components of $G-e$.
Then $\phi(e)$ is a bridge as well.
Let $G'_1$, $G'_2$ be the two components of $G-\phi(e)$.
Then $\phi(V(G_i)) = V(G'_i)$ for all $i\in\set{1,2}$,
or $\phi(V(G_i)) = V(G'_j)$ for all $i,j\in\set{1,2}$,
$i\neq j$.
In either case, $\sep(e)=\sep(\phi(e))$,
and also
$\rel(e,v)=\rel(\phi(e),\phi(v))$ for all $v\in V$.
In total, we have for all $e\in E$ and all $v\in V$:
\begin{align}
\label{eqn:automorphism-sep} \sep(e) & = \sep(\phi(e)) \\
\label{eqn:automorphism-rel} \rel(e,v) & = \rel(\phi(e),\phi(v))
\end{align}
\par
Let $v\in V$.
Then:
\begin{align*}
I_{\phi(v)}(G)
	& = \sum_{e\in E} \rel(e,\phi(v)) \, \Pr{\set{e}} \\
& = \sum_{e\in E} \rel(\phi(e),\phi(v)) \, \Pr{\set{\phi(e)}} 
	& \text{$\phi$ is bijective} \\
& = \sum_{e\in E} \rel(e,v) \, \Pr{\set{\phi(e)}}
	& \text{by~\eqref{eqn:automorphism-rel}} \\
& = \sum_{e\in E} \rel(e,v) \, \Pr{\set{e}}
	& \text{by~\eqref{eqn:automorphism-sep} and symmetric adversary} \\
& = I_v(G)\enspace.
\tag*{\qedhere}
\end{align*}
\end{proof}
\par
The converse of \autoref{prop:symmetric-is-anonymous} does not hold,
as shown in \autoref{fig:anonymous-nonsymmetic}
\vpageref{fig:anonymous-nonsymmetic}.
\autoref{prop:symmetric-is-anonymous} is useful
since symmetry can be recognized
directly from the definitions of
the two special adversaries studied later,
and so we know that they induce anonymous disconnection cost.
\begin{figure}
\centering
\begin{tikzpicture}
 	\node (cl) {};
	\foreach \i in {0,72,144,216,288}{
		\node[player] (cl\i) at ([shift={(\i:1.5cm)}] cl) { };
	};
	\path (cl0) edge (cl72)
		(cl72) edge (cl144)
		(cl144) edge (cl216)
		(cl216) edge (cl288)
		(cl288) edge (cl0);
	\node[player] at (cl0) {$u$};
	\node[player] (c) [right = 3.5cm of cl] {$v$};
 	\node[player] (cr) [right = 3.5cm of c] {};
	\foreach \i in {0,90,180,270}{
		\node[player] (cr\i) at ([shift={(\i:1.5cm)}] cr) { };
	};
	\path (cr0) edge (cr90)
		(cr90) edge (cr180)
		(cr180) edge (cr270)
		(cr270) edge (cr0)
		(cr270) edge (cr)
		(cr90) edge (cr);
	\node[player] at (cr180) {$w$};
	\path (cl0) edge node {$\sfrac{2}{3}$} (c)
		(c) edge node {$\sfrac{1}{3}$} (cr180);
\end{tikzpicture}
\caption{%
	\label{fig:anonymous-nonsymmetic}%
	Let this be the final graph $G=(V,E)$
	and $G_1$ the subgraph to the left starting with~$u$ 
	(\ie the cycle on $5$ vertices),
	and $G_2$ the subgraph to the right starting with $w$.
	Let the adversary assign probabilities, say 
	$\Pr{\set{\set{u,v}}}\df\sfrac{2}{3}$
	and
	$\Pr{\set{\set{v,w}}}\df\sfrac{1}{3}$
	and $0$ to all other edges.
	Then all players in $G_1$ experience the same disconnection cost,
	and the same holds for $G_2$.
	(Precisely, we have 
	$I_x(G) = \sfrac{2}{3} \cdot 6 + \sfrac{1}{3} \cdot 5 = 4+\sfrac{5}{3}$ for all $x\in V(G_1)$
	and $I_y(G) = \sfrac{1}{3} \cdot 6 + \sfrac{2}{3} \cdot 5 = 2 + \sfrac{10}{3}$ for all $y\in V(G_2)$
	and $I_v(G) = \sfrac{2}{3} \cdot 5 + \sfrac{1}{3} \cdot 5 = 5$.)
	The adversary is not symmetric,
	since $\sep(\set{u,v})=5\cdot6=\sep(\set{v,w})$.
	However, disconnection cost is anonymous,
	since an automorphism can only permute players
	within $G_1$ and $G_2$, respectively.}
\end{figure}
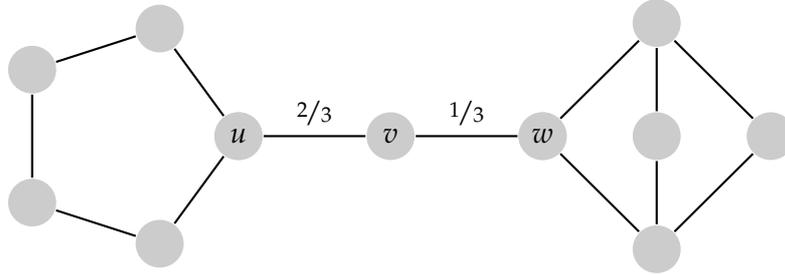%

\section{Related Work}
\label{sec:related-work}%
There is a vast body of literature on game-theoretic network formation,
by far not limited to studies of the price of anarchy.
A~good starting point is the survey by Jackson~\cite{Jac04} from 2004.
We cite several publications below with a bias towards 
studies of the price of anarchy.
In a separate subsection \vpageref{subsec:comparison},
we give a detailed comparison of our model
with work being particularly related to it,
namely~\cite{CFSK04,JW96,BS08a,HS03,HS05,BG00}.
\par\smallskip
Bilateral link formation follows a concept 
given by Myerson~\cite{Mye02} in a different context.
We quote~\cite[\nref{p.}{228}]{Mye02},
emphasis added:
\begin{quote}
\small
Now consider \textbf{a link-formation process in which each player independently writes down 
a list of players with whom 
he wants to form a link},
and the payoff allocation is the fair allocation above
for \textbf{the graph that contains a link 
for every pair of players that have named each other.}
\end{quote}
\par
Jackson and Wolinsky~\cite{JW96} in 1996 introduced the symmetric connections model
and the equilibrium concept of pairwise stability.
The symmetric connections model is best described using the notions of income and payoff.
The income for player $v$ is
$\sum_{\substack{w\in V\\ w\neq v}} \delta^{\dist_{G(S)}(v,w)}$,
where $\delta\in(0,1)$ is a parameter.
Her payoff is income minus building cost.
Note that we have an exponential dependence on distance.
This models to some extent that each link has a probability of $1-\delta$ for failure.
We will elaborate on this later.
\par
Jackson and Wolinsky discussed several variations of \PS,
including what would later be known as \PNE.
We quote~\cite[\nref{p.}{67}]{JW96}:
\begin{quote}
\small
Another possible strengthening of the stability notion would allow for
richer combinations of moves to threaten the stability of a network. Note
that the basic stability notion we have considered requires only that a
network be immune to one deviating action at a time. It is not required
that a network be immune to more complicated deviations, such as a
simultaneous severance of some existing links and an introduction of a new
link by two players [..].
\end{quote}
\par
Watts~\cite{Wat01} in 2001 studied the symmetric connections model
with an extended equilibrium concept:
a graph is considered stable if no player wishes to sell any link
and if no two players wish to establish an additional link
while deleting any number of their links.
Calv\'{o}-Armengol and \.{I}lkili\c{c}~\cite{CI05}
and Corbo and Parkes~\cite{CP05} in 2005 discussed different equilibrium concepts and their relations:
\PNE, \PS,
and proper equilibrium~\cite{Mye02}.
In~\cite{CI05}, among other results,
it was shown that the symmetric connections model has convex cost.
\par
Bloch and Jackson~\cite{BJ07} in 2007 introduced a model with transfers:
each player $v$ decides how much she is willing to pay for a link $\set{v,w}$
or how much she would demand the other endpoint $w$ to pay for the link.
If $v$ offers at least as much as $w$ demands, or vice versa,
the link $\set{v,w}$ is established in the final graph.
Appropriate equilibrium concepts were introduced and discussed.
Bloch and Jackson also compared \PS,
\PNE,
and their transfer model in a separate publication~\cite{BJ06}.
\par
Bala and Goyal~\cite{BG00} in 2000 and in a unilateral setting 
studied a model
where players wish to be connected by a path to as many other players as possible,
but path lengths are unimportant.
They also considered a unilateral version of the symmetric connections model.
In another publication~\cite{BG00a} in the same year,
they extended the first model by allowing each link to fail with a probability $1-p$.
Haller and Sarangi~\cite{HS03,HS05} in 2003 extended this model again
by allowing each link $\set{v,w}$ to fail with its own probability $1-p_{vw}$.
We will elaborate on this later.
\par
Anshelevich, Dasgupta, Tardos, and Wexler~\cite{ADTW08} in 2003
studied the price of anarchy and algorithmic aspects
of a model in which each player has a set of terminals 
and aims to construct a network which connects her terminals.
For a related model,
Anshelevich, Dasgupta, Kleinberg, Tardos, Wexler, and Roughgarden~\cite{ADK+04} in 2004
studied the price of stability.
Also in 2004, Christin and Chuang~\cite{CC04} studied a model for network formation
with an extended cost function modeling peer-to-peer networks,
and Christin, Grossklags, and Chuang~\cite{CGC04}
looked at it under the aspect of different game-theoretic principles.
\par
Chun, Fonseca, Stoica, and Kubiatowicz~\cite{CFSK04} in 2004
experimentally studied an extended version of the sum-distance model.
\par
Johari, Mannor, and Tsitsiklis~\cite{JMT06} in 2006
studied a model in which each vertex wishes to send a given amount of traffic 
to some of the other vertices,
and only cares whether the traffic eventually arrives at the destination.
There is a handling cost at each vertex,
which is proportional to the amount of traffic through that vertex.
\par\smallskip
The work of Fabrikant, Luthra, Maneva, Papadimitriou, and Shenker \cite{FLM+03}
from 2003 is to the best of the author's knowledge the first quantitative study of 
the price of anarchy in a model that fits into the framework considered here,
as per \autoref{sec:model-framework}.
They considered the unilateral sum-dis\-tance model
and proved a bound of $\max\set{1,\,O(\sqrt{\alpha})}$ on the price of anarchy in general,
and an $O(1)$ bound for $\alpha>\frac{(n-1)\,n}{2}$.
They conjectured that for $\alpha=\Omega(1)$,
all non-transient \NE\ were trees
-- the Tree Conjecture.
A \NE\ is called \term{transient}
when there exists a sequence of strategy changes 
in which each player changing her strategy maintains her individual cost,
and finally a strategy profile is reached which is not a \NE\ anymore.
The Tree Conjecture was based on the observation
that all \NE\ constructed so far at that time, for $\alpha>2$,
were trees or transient ones (namely the Petersen graph for $\alpha\leq 4$).
The Tree Conjecture was later, in 2006, disproved by
Albers, Eilts, Even-Dar, Mansour, and Roditty~\cite{AEE+06}
by showing that for each $n_0$,
there exists a non-transient \NE\ on $n\geq n_0$ vertices 
containing cycles, 
for any $1<\alpha\leq\sqrt{\sfrac{n}{2}}$.
\par
Corbo and Parkes~\cite{CP05} in 2005 considered the bilateral version of the sum-distance model.
They showed an $O(\sqrt{\alpha})$ bound for \mbox{$1\leq\alpha<n^2$} on the price of anarchy.
As noticed later in 2007 by Demaine et al.~\cite{DHMZ07}, the proof in fact yields
$O(\min\set{\sqrt{\alpha},\,\sfrac{n}{\sqrt{\alpha}}})$.
\par
Albers et al.~\cite{AEE+06} in 2006 not only disproved the Tree Conjecture,
but also improved the bounds on the price of anarchy for the unilateral sum-distance model:
they gave constant upper bounds for $\alpha =O(\sqrt{n})$ 
and $\alpha\geq 12 n \sceil{\log n}$,
as well as an upper bound
for any $\alpha$ of
\begin{equation*}
15 \, \parens{1+\parens{\min\set{\sfrac{\alpha^2}{n},\,\sfrac{n^2}{\alpha}}}^{\sfrac{1}{3}}}
\enspace.
\end{equation*}
An $O(1)$ upper bound for $\alpha=O(\sqrt{n})$ was also independently proved by Lin~\cite{Lin03}.
These bounds were again improved by
Demaine, Hajiaghayi, Mahini, and Zadimoghaddam~\cite{DHMZ07} in 2007.
They showed a bound of $2^{O(\sqrt{\log n})}$ for any $\alpha$
and a constant bound for $\alpha=O(n^{1-\eps})$ for any constant $\eps>0$.
For the bilateral version,
they proved the $O(\min\set{\sqrt{\alpha},\,\sfrac{n}{\sqrt{\alpha}}})$ bound of Corbo and Parkes tight.
Recently, in 2010, Mihal{\'a}k and Schlegel~\cite{MS10a}
proved that for the unilateral sum-distance model and $\al \geq 273 n$,
all equilibria are trees,
which implies a constant bound on the price of anarchy in that range of~$\al$.
\par
Moscibroda, Schmid, and Wattenhofer~\cite{MSW06} in 2006
studied the price of anarchy in a variation of 
the sum-distance model
where the distance between two vertices is generalized,
that is, it may be given by any metric.
The cost function uses the stretch,
that is the actual distance in the constructed graph
divided by the distance that a direct connection would provide.
Halevi and Mansour~\cite{HM07} in 2007 studied the price of anarchy in the sum-distance model
under the generalization that each player has a list of \enquote{friends},
that is, a list of other vertices and she is only interested in her distance to those.
Demaine et al.~in~\cite{DHMZ07} in 2007 also considered the max-distance model:
indirect cost for $v$ is $\max_{w\in V}\dist(v,w)$.
Upper bounds were shown for \ULF\
and tight bounds for \BLF.
For \ULF, improved bounds were recently shown in~\cite{MS10a}.
\par
Brandes, Hoefer, and Nick~\cite{BHN08} in 2008 studied
a variant of the sum-distance model assigning a finite distance to
pairs of disconnected players, allowing for disconnected equilibria.
They proved structural properties and bounds on the price of anarchy.
Laoutaris, Poplawski, Rajaraman, Sundaram, and Teng \cite{LPR+08} in 2008
considered \term{bounded budget connection games},
a variant of the sum-distance model with player-dependent link costs,
lengths, and preferences $w(u,v)$ expressing the importance 
for player $u$ of having a good connection to player $v$,
and finally a budget for each player limiting the number of links that this player can build.
They considered existence of equilibria and proved bounds on the price of anarchy and stability.
An important special case is the \term{uniform} version,
which has link costs, link lengths, and preferences all equal,
and also all players have the same limit on their budget.
Recently, this uniform version was also studied by Demaine and Zadimoghaddam~\cite{DZ10}.
They proved a tight upper bound
and, more importantly, showed how to induce equilibria with small social cost.
They used a technique called \term{public service advertising},
previously studied for different games by Balcan, Blum, and Mansour~\cite{BBM09}.
\par
Baumann and Stiller~\cite{BS08a} in 2008
considered the price of anarchy in the symmetric connections model.
Demaine et al.~\cite{DHMZ09} in 2009 studied the price of anarchy 
in a cooperative variant of the sum-distance model.
They also looked at the case that
links can only be formed for certain pairs of vertices,
that is, the underlying \enquote{host} graph needs not to be a complete one.

\subsection*{Comparison of Our Model with Related Work}
\label{subsec:comparison}%
Our adversary model addresses robustness in a way
that -- to the best of the author's knowledge -- has not been studied theoretically before.
We compare our approach to previous work that also addresses robustness.
\par
Chun, Fonseca, Stoica, and Kubiatowicz~\cite{CFSK04} 
experimentally studied an extended version of the sum-distance model
and considered robustness.
To simulate failures, they removed some vertices randomly.
To simulate attacks, they removed vertices starting with those 
having highest degree.
\par
The symmetric connections model of Jackson and Wolinsky~\cite{JW96}
can also be interpreted from a robustness point-of-view.
Recall that in the symmetric connections model there is a parameter $\delta\in(0,1)$,
and payoff $\payoff_v(S)$ for player $v$ under strategy profile $S$ is defined
\begin{equation*}
\payoff_v(S) \df \sum_{\substack{w\in V\\ w\neq v}} \delta^{\dist_{G(S)}(v,w)} - \nrm{S_v} \, \alpha\enspace.
\end{equation*}
\par
An interpretation is that $v$ receives one unit of income from 
each other vertex $w$ along a shortest path between $v$ and $w$.
However, each link has a probability $1-\delta$ of failure,
so the expected income from $w$ is the probability that none 
of the $\dist_{G(S)}(v,w)$ links fails, 
which is $\delta^{\dist_{G(S)}(v,w)}$ if we assume stochastic independence of failures.
For \BLF,
Baumann and Stiller~\cite{BS08a} gave an expression for the exact price of anarchy
for $\alpha\in(\delta-\delta^2,\delta-\delta^3)$,
which implies an $O(1)$ bound
(the constant being bounded by $\frac{4}{1+2\delta}$).
The price of anarchy is $1$ for $\alpha<\delta-\delta^2$, 
following from~\cite{JW96}.
The price of anarchy in the range $\alpha>\delta-\delta^3$ is not fully understood yet.
\par\medskip\noindent
The symmetric connections model is different from ours in many respects:
\begin{itemize}
\item All links have the same probability of failure.
In our model, links can have different probabilities,
and these may even depend on the final graph.\footnote{%
However, large parts of our analysis will be restricted to two specific cases:
one in which the adversary picks a link uniformly at random (simple-minded adversary)
and another in which he picks a link that causes maximum overall damage (smart adversary).}
\item The failure of a link $e$ and the failure of a link $f$ are independent events for $e\neq f$,
at least along the concerned paths.
In our model, the failures of $e$ and $f$ are mutually exclusive events.
\item Alternative paths are not considered; 
it is assumed that routing happens along a specific shortest path
that is fixed before the random experiment that models the link failures is conducted.
In our model, \emphasis{all} paths are considered.
However, we do not consider path lengths.
\end{itemize}
\par
Bala and Goyal~\cite{BG00a} studied a variation of the symmetric connections model,
which is closer to ours.
In their model, each vertex receives an amount of $1$ from each
vertex it is connected to via some path.
Each link has a probability $1-p$ of failure, 
$p\in[0,1]$ being the same for all links and independent of the final graph.
Failures of two distinct links are stochastically independent.
Income of a vertex $v$ is the expected number of vertices
to which $v$ is connected via a path.
Unilateral link formation is used.
They considered structural properties of optima and \NE,
in particular pointing out cases where \NE\ are \enquote{super-connected},
\ie connected and not containing bridges.
They also showed that for some regions of parameters,
there exist \NE\ that are also optima (\ie they show a price of stability of $1$ for these regions).
\par
Haller and Sarangi~\cite{HS03,HS05} studied an extension of
the model of Bala and Goyal~\cite{BG00a}.
In their model, each link $\set{v,w}$ may fail with its \emphasis{own} probability \mbox{$1-p_{vw}$.}
They also considered structural properties of optima and \NE\
as well as relations of optima and \NE,
including the price of stability similar to~\cite{BG00a}.
Like the symmetric connections model, 
their model shows several differences to ours:
\begin{itemize}
\item The failure probability of each link $\set{v,w}$ is $1-p_{vw}$,
independent of the final graph.\footnote{%
Haller and Sarangi also briefly discussed
failure probabilities depending on the final graph.
They considered an example where for non-increasing functions $f_v(\cdot)$ and parameters $P_{vw}$
the probabilities are defined $p_{vw}(S)\df f_v(\deg_{G(S)}(v)) \, f_w(\deg_{G(S)}(w)) \, P_{vw}$
if $v$ and $w$ have a link between them, and $0$ otherwise.}
In our model, failure probabilities depend on the final graph.
\item Failures of two different links are stochastically independent.
In our model, they are mutually exclusive events.
(This difference is exactly as between the symmetric connections model and ours.)
\end{itemize}
\par\smallskip
Generally, independent link failures model 
the unavailability of links due to, \eg
deterioration, maintenance times,
or influences affecting the whole infrastructure or large parts of it
(\eg natural disasters).
Our adversary model, on the other hand,
models the situation when faced with an entity
that is malicious but only has limited means 
so that it can only destroy a limited number of links
(we limit this number to~$1$ in this work).

\section{The Bridge Tree}
\label{sec:bridge-tree}%
We conduct some preparations for the analysis of equilibria
in our adversary model,
which will be useful
regardless of the link formation rule and the equilibrium concept.
In the end, in \autoref{lem:double-count},
we will have developed 
a simple method to bound the sum of relevances $\Rel(v)$
for each player $v$, which will later help to bound the disconnection cost.
It will be helpful in several places to
consider a variation of the block graph,\footnote{%
See, \eg~\cite[\nref{p.}{56}]{Die05} for the definition of the block graph.}
which we call the \term{bridge tree}.
Its definition requires some preparation.
If $W\subseteq V$ is maximal under the condition
that the induced subgraph $G[W]$ is connected and does not contain any bridges of $G[W]$,
we call $W$ a \term{bridgeless connected component}, abbreviated \enquote{\bcc}.
The proof of the following proposition is straightforward.
\begin{proposition}
\label{prop:bridge-tree-definition}%
A set of vertices $W \subseteq V$ is a \bcc\ 
if and only if $W$ is maximal under the condition
that the induced subgraph $G[W]$ is connected and does not contain any bridges of $G$.
\end{proposition}
\begin{proof}
Let $W$ be maximal under the condition of $G[W]$ being connected
and not containing any bridges of $G[W]$,
\ie we follow the original definition given above.
Clearly, $G[W]$ does not contain any bridges of~$G$,
since if removal of some edge disconnects $G$, 
then it also disconnects $G[W]$ if the endpoints of this edge are in $W$.
We choose $U\supseteq W$ maximal under the condition that $G[U]$ is connected
and $G[U]$ does not contain any bridges of $G$.
Suppose $U\neq W$.
Then $G[U]$ contains a bridge $e$ of $G[U]$.
Since this is no bridge of $G$, it is located on a cycle $C$.
Then $V(C)\nsubseteq U$, since $e$ is a bridge of $G[U]$.
But $G[U\cup V(C)]$ would still be connected and would contain no bridge of $G$.
This contradicts the maximality of~$U$.
\par
Now let $W$ be maximal under the condition of $G[W]$ being connected
and not containing any bridges of $G$.
If $G[W]$ contained a bridge $e$ of $G[W]$ (but not of $G$),
we could use the cycle-argument from before to augment $W$ 
and have a contradiction to its maximality.
Suppose there is $U\supsetneq W$ such that $G[U]$ is connected
and $G[U]$ does not contain any bridges of $G[U]$.
Then $G[U]$ contains a bridge of $G$.
As noted earlier, this is also a bridge of $G[U]$, a contradiction.
\end{proof}
What we call \enquote{\bcc} is sometimes called \enquote{block} in the literature,
and what we call \enquote{bridge tree} is then called \enquote{bridge-block tree}.
We refrain from using \enquote{block} here,
since it usually is related to \emphasis{vertex-}connectivity;
see, \eg~\cite[\nref{p.}{55}]{Die05}.
\par
Every vertex is contained in exactly one \bcc.
If $W$ is a \bcc,
we have to remove at least $2$ edges from $G[W]$ 
in order to make it disconnected.
A graph from which we have to remove at least $2$ edges to make it disconnected
is also called being \enquote{$2$-edge-connected} in common terminology,
provided that it has more than $1$ vertices;
see, \eg~\cite[\nref{p.}{12}]{Die05}.
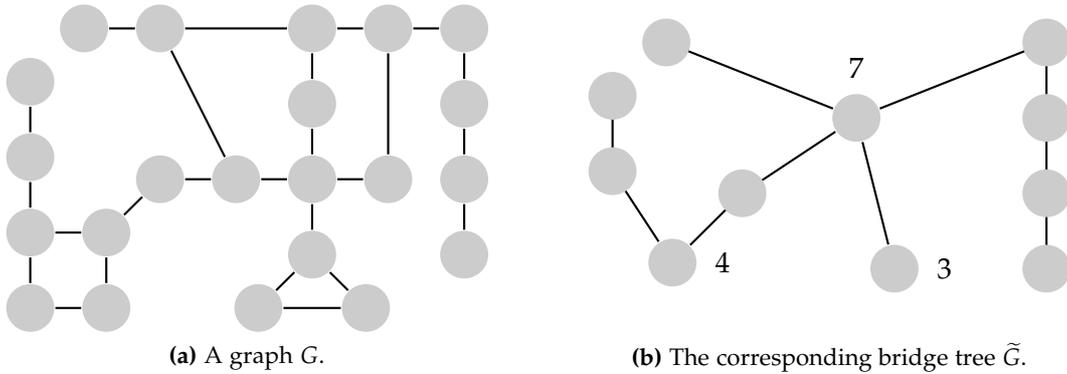
\begin{figure}
\newcommand{\skel}{%
	\node (a1) { };
	\node (a2) [below = of a1] { };
	\node (a3) [left = of a2] { };
	\node (a4) [above = of a3] { };
	\node (b1) [above right = of a1] { };
	\node (b2) [above = 2cm of b1] { };
	\node (b3) [right = 2cm of b2] { };
	\node (b4) [right = of b3] { };
	\node (b5) [below = 2cm of b4] { };
	\node (b6) [left = of b5] { };
	\node (b7) [left = of b6] { };
	\node (b8) [below = of b3] { };
	\node (c1) [left = of b2] { };
	\node (d1) [above = of a4] { };
	\node (d2) [above = of d1] { };
	\node (e1) [below = of b6] { };
	\node (e2) [below left = of e1] { };
	\node (e3) [below right = of e1] { };
	\node (f1) [right = of b4] { };
	\node (f2) [below = of f1] { };
	\node (f3) [below = of f2] { };
	\node (f4) [below = of f3] { };}
\centering
\subfloat[]%
[\label{fig:bridge-tree:1}A graph $G$.]%
{\begin{tikzpicture}[every node/.style=player,node distance=1cm]
	\skel
	\path (a1) edge (a2) (a2) edge (a3) (a3) edge (a4) (a4) edge (a1);
	\path (a1) edge (b1);
	\path (b2) edge (b3)
		(b3) edge (b4)
		(b4) edge (b5)
		(b5) edge (b6)
		(b6) edge (b7)
		(b7) edge (b1);
	\path (b7) edge (b2);
	\path (b3) edge (b8) (b8) edge (b6);
	\path (b2) edge (c1);
	\path (d1) edge (d2);
	\path (b6) edge (e1);
	\path (e1) edge (e2) (e2) edge (e3) (e3) edge (e1);
	\path (a4) edge (d1);
	\path (b4) edge (f1);
	\path (f1) edge (f2) 
		(f2) edge (f3)
		(f3) edge (f4);
\end{tikzpicture}}
\hfill
\subfloat[]%
[\label{fig:bridge-tree:2}The corresponding bridge tree $\tG$.]%
{\begin{tikzpicture}[node distance=1cm]
	\skel
	\begin{scope}[every node/.style=player]
	\node (B) [above right = 1cm and 1.5 cm of b1, label=above:{$7$}] {};
	\node (E) [below = of b6, label=right:{$3$}] {};
	\node (A) [below left = 1.3 cm of b1, label=right:{$4$}] {};
	\end{scope}
	\path (A) edge (d1) (d1) edge (d2);
	\path (A) edge (b1) (b1) edge (B);
	\path (B) edge (c1);
	\path (B) edge (E);
	\path (B) edge (f1)
		(f1) edge (f2)
		(f2) edge (f3)
		(f3) edge (f4);
	\begin{scope}[every node/.style=player]
	\node at (b1) { };
	\node at (c1) { };
	\node at (d1) { };
	\node at (d2) { };
	\node at (f1) { };
	\node at (f2) { };
	\node at (f3) { };
	\node at (f4) { };
	\end{scope}
\end{tikzpicture}}
\caption{%
	\label{fig:bridge-tree}%
	Bridge tree construction.
	Vertices representing \bcc s of more than $1$ vertices 
	have their number of vertices attached, here $4$, $7$, and $3$,
	respectively.}
\end{figure}%
\par
Now we introduce the bridge tree of a graph $G=(V,E)$.
It is the graph $\tG=(\widetilde{V},\widetilde{E})$ defined by:
\begin{align*}
\widetilde{V}&\df\setst{B\subseteq V}{\text{$B$ is a \bcc}}\enspace,\\
\widetilde{E}&\df
\setst{\set{B,B'}}{B,B'\in \widetilde{V} \land
\exists v\in B, w\in B': \set{v,w}\in E}\enspace.
\end{align*}
Then $\tG$ is a tree (assuming $G$ is connected).
By \autoref{prop:bridge-tree-definition},
if $\phi$ maps each vertex of $G$ to the \bcc\ in which it is contained,
then $\set{v,w} \mapsto \set{\phi(v),\phi(w)}$
maps from the set of bridges of $G$
to the set of edges of $\tG$ and is bijective.
We make the following special convention concerning the bridge tree:
\begin{convention*}
\label{con:vertex-counting}%
Whenever we speak of the number of vertices in a subgraph $T$ of the bridge tree,
we count $\card{B}$ for each vertex $B\in V(T)$.
\end{convention*}
\par
In other words, we count the vertices that would be there
if we expanded $T$ back to
its corresponding subgraph of $G$.
\autoref{fig:bridge-tree} 
shows an example.
Since each vertex of $G$ is in exactly one \bcc,
counting in this way for $\widetilde{V}$ yields the number of vertices in $G$, \ie~$n$.
\par
On several occasions, when considering the effect of building additional edges,
we treat vertices of the bridge tree as players.
This is justified since edges inside \bcc s have relevance $0$.
Hence for a strategy profile $S$
and $B,B'\in \tV$ 
the effect in disconnection cost of a new edge between a player from $B$
and a player of $B'$ is specific to the pair $\set{B,B'}$
and not to the particular players.
\par
For a path $P$~in $G$, 
let $\tP$ be its contracted counterpart in $\tG$,
\ie we replace in $P$ each maximal sequence of vertices 
from the same \bcc\ $B\in\widetilde{V}$ with $B$.
Then the length $\length{\tP}$ of $\tP$ 
is the number of bridges in $P$.
For each pair $v,w\in V$ denote $P(v,w)$ an arbitrary shortest path from $v$ to $w$;
and $\shp(v)\df\setst{P(v,w)}{w\in V}$.
The bridge tree helps bounding the disconnection cost.
We conclude this section with a preparation for this.
For each $v\in V$ and $e\in E$ we easily observe:
\begin{equation}
\label{eqn:rel-paths}%
\rel(e,v)=\begin{cases}
0 & \text{if $e$ is a non-bridge}\\
\card{\setst{P\in\shp(v)}{e\in E(P)}} & \text{if $e$ is a bridge.}
\end{cases}
\end{equation}
We use this to prove the following lemma,
which we will apply in 
\autoref{sec:uni-simple} and \autoref{sec:bilateral}.
The lemma relates relevance to path length, and so to diameter.
This is possible since it is the same
to count for each edge the number of paths that cross this edge
(establishing the connection to relevance)
as to count for each path the number of its edges
(establishing the connection to path length).
\begin{lemma}
\label{lem:double-count}%
For each $v\in V$ we have
$R(v)
\leq (n-1) \, \diam(\tG)$.
\end{lemma}
\begin{proof}
Fix $v\in V$.
We have
\begin{align*}
R(v) & \stackrel{\text{def}}{=}
\sum_{e\in E} \rel(e,v)
\stackrel{\text{\eqref{eqn:rel-paths}}}{=}
\sum_{\substack{e\in E\\\text{$e$ is a bridge}}} \card{\setst{P\in\shp(v)}{e\in E(P)}}\\
& = \sum_{P\in \shp(v)} \card{\setst{e\in E(P)}{\text{$e$ is a bridge}}}
= \sum_{P\in \shp(v)} \length{\tP}
\leq (n-1) \, \diam(\tG)\enspace.
\end{align*}
The last estimation is true
since the bridge tree is a tree and so
every path is a shortest path.
\end{proof}

\section{A Simple Bound on the Price of Anarchy}
\label{sec:simple-bound}%
We give an upper bound on the price of anarchy 
for a general adversary.
It holds independently of the link formation rule and the equilibrium concept,
provided that equilibria have few edges.
\begin{proposition}
\label{prop:simple-bound}%
Let $S$ be any strategy profile
and (as usual) $m=\card{E(S)}$.
\begin{enumerate}[label=(\roman*)]
\item \label{prop:simple-bound:ii}%
If $m=O(n)$, then $\frac{C(S)}{\OPT(n,\al)} = O(1+\frac{n}{\alpha})$.
\item \label{prop:simple-bound:iii}%
If $m=O(n)$ and $\alpha=\Omega(n)$, 
then $\frac{C(S)}{\OPT(n,\al)} = O(1)$.
\end{enumerate}
\end{proposition}
\begin{proof}
Since $\sep(e)=O(n^2)$ for all $e$, we have
\begin{equation*}
\SC(S)
=O\parens{m \, \alpha + n^2 \, \sum_{e\in E}\Pr{\set{e}}}
=O(n \alpha + n^2)
\enspace.
\end{equation*}
Since an optimum is connected,
the optimal social cost is $\Omega(n\alpha)$.
Dividing by this yields \ref{prop:simple-bound:ii}.
Assertion \ref{prop:simple-bound:iii} follows from \ref{prop:simple-bound:ii}.
\end{proof}
The following corollary is obvious.
\begin{corollary}
Fix any link formation rule and equilibrium concept.
\begin{enumerate}[label=(\roman*)]
\label{cor:simple-bound}%
\item \label{cor:simple-bound:ii}%
If the number of edges in each equilibrium is $O(n)$,
then the price of anarchy is $O(1+\frac{n}{\alpha})$.
\item \label{cor:simple-bound:iii}%
If the number of edges in each equilibrium is $O(n)$
and moreover $\alpha=\Omega(n)$, then the price of anarchy is $O(1)$.
\end{enumerate}
\end{corollary}
\par
A remark on the meaning of $O$ and $\Om$ is in order.
Recall that we use this notation to avoid having to
introduce all occurring constants explicitly,
and that a \term{constant} is required to be independent of all game parameters,
strategy profiles, etc.
For example, 
the constant hidden in the premise \enquote{$m=O(n)$} 
in \autoref{prop:simple-bound}\ref{prop:simple-bound:ii} 
is required to be independent of $n$, $\al$, and $S$,
while the constant hidden in the conclusion 
\enquote{$\frac{C(S)}{\OPT(n,\al)} = O(1+\frac{n}{\alpha})$}
is guaranteed to have the same independence.
The proof reflects that this is true.
In \autoref{cor:simple-bound}\ref{cor:simple-bound:ii},
it is required that there exists a constant $c>0$
such that for each equilibrium $S$ we have $\card{E(S)} \leq cn$.
Here as well, $c$ is required to be independent of $n$, $\al$, and $S$.
\par
The main goal of \autoref{sec:uni-simple} and \autoref{sec:uni-smart}
is to show a bound of $O(1)$ on the price of anarchy for \ULF\ and restricted to two special adversaries,
which are chosen to mark extreme cases.
We will there proceed in showing the $O(n)$ bound on the number of edges in a \NE\ first.
In \autoref{sec:uni-simple}, we will then bound disconnection cost of \NE\ by $O(n\al)$.
In \autoref{sec:uni-smart}, we achieve the same bound 
under the condition that $\al < c n$ for a constant specified there.
If $\al \geq cn$, then we are done by \autoref{cor:simple-bound}\ref{cor:simple-bound:iii}.
\par
The $O(1)$ bound would follow trivially
if we could show an $O(n\alpha)$ bound
for the social cost of any strategy profile.
However, later \autoref{prop:sc-path} shows
that there is no hope for this, 
and hence we will have to exploit characteristics of \NE\
in order to prove our bounds.

\section{Optima, Nash Equilibria, Price of Stability}
\label{sec:uni-opt-ne-stability}%
We stick to \ULF\ (and \NE\ as equilibrium concept) in this
and the following two sections.
The aim of this section is to
construct optima and \NE,
and finally to show how a bound on the price of stability follows easily.
The adversaries considered are a general one,
\ie without any additional assumptions,
and one inducing anonymous disconnection cost.
Although we do everything for \ULF\ and \NE\ here,
most of the results can be carried over to \BLF\ and \PNE\ or \PS,
as will be discussed in \autoref{sec:bilateral}.
\begin{proposition}
\label{prop:uni-opt}%
An optimum has social cost $\Theta(n\alpha)$.
More precisely:
\begin{enumerate}[label=(\roman*)]
\item If $\alpha\leq 2\,(n-1)$, the cycle is an optimum;
it has social cost $n\alpha$.
\item If $\alpha\geq 2\,(n-1)$, a star is an optimum;
it has social cost $(n-1)\,(\alpha+2)$.
\end{enumerate}
\end{proposition}
\begin{proof}
An optimum can only be the cycle or a tree,
because any graph containing a cycle has already the building cost $n\alpha$ of the cycle,
and the cycle has optimal disconnection cost.
So an optimum is either the cycle, or it is cycle-free.
Let $T$ be any tree.
We have its indirect cost:
\begin{align*}
\MoveEqLeft
\sum_{e\in E(T)} \sep(e) \, \Pr{\set{e}}
=2\sum_{e\in E(T)}\nn(e)\,(n-\nn(e)) \, \Pr{\set{e}}\\
& \geq2\cdot1\,(n-1)\sum_{e\in E(T)}\Pr{\set{e}}
=2\,(n-1)\enspace.
\end{align*}
Hence the social cost of a tree is at least
\begin{equation*}
(n-1)\,\alpha+2\,(n-1)=(n-1)\,(\alpha+2)\enspace.
\end{equation*}
The social cost of the cycle is $n\alpha$.
So if $\alpha\leq 2\,(n-1)$, the cycle is better or as good as any tree,
hence it is an optimum.
If $\alpha > 2\,(n-1)$, then we look for a good tree.
A star has social cost $(n-1)\,(\alpha+2)$, 
which matches the lower bound given above,
and is hence optimal (and better than the cycle).
However, for $\alpha = 2\,(n-1)$,
both cycle and star are optimal with social cost $2n\,(n-1)$.
\end{proof}
\par
The following simple remark later will help establishing concrete bounds
on the price of anarchy.
\begin{remark}
\label{rem:concrete-bound-general}%
Assume there are constants $c_0, c_1 > 0$ such that
the social cost of all equilibria is bounded by $(c_1 n + c_0)\,\al$.
Then the price of anarchy is bounded by $c_1 + \frac{c_1 + c_0}{n-1}$.
\end{remark}
\begin{proof}
If the optimum is $n \al$,
we have the ratio $\frac{(c_1 n + c_0)\,\al}{n\al} = c_1 + \frac{c_0}{n} 
\leq c_1 + \frac{c_1 + c_0}{n-1}$.
Otherwise, if the optimum is $(n-1)\,(\al+2)$,
we have the ratio 
$\frac{(c_1 n + c_0)\,\al}{(n-1)\,(\al+2)} 
< \frac{c_1 n + c_0}{n-1} 
= \frac{c_1 \, (n-1) + c_1 + c_0}{n-1} 
= c_1 + \frac{c_1 + c_0}{n-1}$.
\end{proof}
\par
The following two propositions can be proved by
appropriate cost-benefit analysis.
\begin{proposition}
\label{prop:uni-ne-star}%
Let $S$ be a star
with edges pointing outward.
\begin{enumerate}[label=(\roman*)]
\item\label{prop:uni-ne-star:i}%
If $\alpha\geq n-1$, then $S$ is a \NE.
\item\label{prop:uni-ne-star:ii}%
If $\alpha\geq 2-\frac{1}{n-1}$, 
then $S$ is a \NE\ if disconnection cost is anonymous.
\end{enumerate}
In both cases, strict inequality implies a \MNE.
\end{proposition}
\begin{proof}
\ref{prop:uni-ne-star:i}~Since all edges point outward,
the center player is the only one
that could sell edges,
but this would make the graph disconnected.
Exchanges of edges by the center cannot lead to a different strategy profile.
The maximum disconnection cost 
is experienced by a leaf vertex
when the probability measure is concentrated
on the one edge that connects it to the rest.
The disconnection cost is then $n-1$.
Since this is at most $\alpha$,
there is no incentive to buy additional edges.
Hence no player can strictly improve her individual cost
by changing her strategy.
(If multiple edges between the same players were allowed,
the center could build additional edges.
However, since $\al > 1$, this would increase her cost.)
\par
\ref{prop:uni-ne-star:ii}~Disconnection cost of the center is $1$.
By anonymity,
all leafs experience the same disconnection cost.
It follows easily from this
that all edges have the same probability, namely $\frac{1}{n-1}$.
Disconnection cost of a leaf hence is
\begin{equation*}
\frac{(n-1) + (n-2)}{n-1}
=\frac{(n-1) + (n-1) - 1}{n-1}
= 2 - \frac{1}{n-1}\enspace.
\end{equation*}
Now we apply the same arguments as for part~\ref{prop:uni-ne-star:i}.
Maximality is clear in both cases by
\autoref{rem:mne}.
\end{proof}
\begin{proposition}
\label{prop:uni-ne-cycle}%
Let $S$ be a cycle
with all edges pointing in the same direction
(either all clockwise or all counter-clockwise).
\begin{enumerate}[label=(\roman*)]
\item\label{prop:uni-ne-cycle:i}%
If $\alpha\leq 1$, then $S$ is a \MNE.
\item\label{prop:uni-ne-cycle:ii}%
If $\alpha\leq \frac{1}{2}\sfloor{\frac{n-1}{2}}$, 
then $S$ is a \MNE\ if disconnection cost is anonymous.
\end{enumerate}
\end{proposition}
\begin{proof}
Maximality in both cases is due to the cycle
having minimum disconnection cost, namely $0$.
Buying or exchanging edges is also not beneficial
since the cycle already has minimum disconnection cost.
We only have to check whether it is beneficial 
for a player $v$ to sell her one edge.
Selling the edge yields a path with $v$ at one of its ends.
This increases disconnection cost for $v$ to at least $1$,
since the removal of any edge disconnects $v$ from at least one other vertex.
This proves~\ref{prop:uni-ne-cycle:i}.
\par
To prove~\ref{prop:uni-ne-cycle:ii},
we have to establish a better lower bound on
the new disconnection cost for~$v$.
Anonymity of indirect cost allows us to do so.
Let the path be $(v_1,e_1,v_2,\hdots,e_{n-1},v_n)$ with $v=v_1$.
We claim that
\begin{equation}
  \label{eqn:symmetry-path}%
  \Pr{\set{e_i}} = \Pr{\set{e_{n-i}}} \quad \text{for all $i\in\set{1,\hdots,n-1}$}
  \enspace,
\end{equation}
\ie the adversary behaves like a symmetric one on this graph.
From~\eqref{eqn:symmetry-path} it follows 
\mbox{$\sum_{i=1}^{\sceil{\sfrac{(n-1)}{2}}} \Pr{\set{e_i}} \geq \frac{1}{2}$}.
Since each of the edges $e_1,\hdots,e_{\sceil{\sfrac{(n-1)}{2}}}$
has relevance at least $\sfloor{\sfrac{(n-1)}{2}}$
for $v$, the proposition follows.
\par
We are left with proving~\eqref{eqn:symmetry-path}.
For each $i \in \setn{n}$ we write $I_i$ for the indirect cost of vertex $v_i$, 
and moreover define its \term{left indirect cost} by
$I_i^\Left \df \sum_{j=1}^{i-1} j \, \Pr{\set{e_j}}$
and its \term{right indirect cost} by
$I_i^\Right \df \sum_{j=1}^{n-i} j \, \Pr{\set{e_{n-j}}}$.
Then clearly $I_i = I_i^\Left + I_i^\Right$.
It suffices to show~\eqref{eqn:symmetry-path} for $i < \frac{n}{2}$.
We have (in fact even for $i \leq n-1$) on the one hand:
\begin{align*}
  I_i & = I_i^\Left + (n-i) \, \Pr{\set{e_i}} + I_{i+1}^\Right \\
  & = I_i^\Left + (n-i) \, \Pr{\set{e_i}} + I_{i+1} - I_{i+1}^\Left \\
  & = I_i^\Left + (n-i) \, \Pr{\set{e_i}} + I_{i+1} - \parens{I_i^\Left + i \, \Pr{\set{e_i}}} \\
  & = (n-2i) \, \Pr{\set{e_i}} + I_{i+1}\enspace.
\end{align*}
On the other hand, 
along the same lines we prove $I_{n-i+1} = (n-2i) \, \Pr{\set{e_{n-i}}} + I_{n-i}$.
By anonymity, $I_i = I_{n-i+1}$ and $I_{i+1} = I_{n-i}$.
It follows $(n-2i) \, \Pr{\set{e_i}} = (n-2i) \, \Pr{\set{e_{n-i}}}$,
and since $i < \frac{n}{2}$, this means $\Pr{\set{e_i}} = \Pr{\set{e_{n-i}}}$.
\end{proof}
For anonymous disconnection cost,
this proves existence of \NE\
for all ranges of $\alpha$
provided that $n\geq 9$,
since then $2-\frac{1}{n-1}\leq2\leq\frac{1}{2}\sfloor{\frac{n-1}{2}}$.
In the range 
$2-\frac{1}{n-1}\leq\alpha\leq\frac{1}{2}\sfloor{\frac{n-1}{2}}$
two very different topologies -- namely cycle and star --
co-exist as \NE.
\begin{convention*}
All our statements on upper bounds on the price of anarchy
are restricted to those combinations of $n$ and $\al$ for which
equilibria exist for the respective adversary.
Instead of this convention, we could rely on the maximum over
the empty set being defined to~$-\infty$.
Hence any alleged upper bound on the price of anarchy would be true
in case that no equilibria exist.
\end{convention*}
\par\medskip
The following is a consequence of 
\autoref{prop:uni-opt}, \autoref{prop:uni-ne-star},
and\linebreak \autoref{prop:uni-ne-cycle}.
\begin{theorem}
\label{thm:uni-stability}%
For anonymous disconnection cost and $n \geq 9$ the price of stability is bounded by $1+\frac{8}{n-2} = 1+o(1)$.
\end{theorem}
\begin{proof}
For $\alpha\leq\frac{1}{2}\sfloor{\frac{n-1}{2}}$ the cycle is a \NE\
as well as an optimum, and so the price of stability is $1$.
For $\alpha\geq 2\,(n-1)$ a star is a \NE\
as well as an optimum, and so the price of stability is $1$.
\par
For $\frac{1}{2}\sfloor{\frac{n-1}{2}}\leq\alpha\leq 2\,(n-1)$,
the star is a \NE\ (since $2 \leq \al$ by $n\geq 9$, so $2-\frac{1}{n-1} < 2 \leq \al$)
and the cycle is an optimum.
The price of stability so is upper-bounded by
\begin{equation*}
\frac{(n-1)\,(\alpha+2)}{n\alpha}
\leq 1+\frac{2}{\alpha}
\leq 1+\frac{4}{\sfloor{\frac{n-1}{2}}}
\leq 1+\frac{8}{n-2}
\enspace.\qedhere
\end{equation*}
\end{proof}

\section{Simple-Minded Adversary}
\label{sec:uni-simple}%
The simple-minded adversary picks an edge uniformly at random,
that is, $\Pr{\set{e}}=\frac{1}{m}$ for all $e\in E$.
Then we have individual and social cost:
\begin{align*}
C_v(S) & = \nrm{S_v}\,\alpha + \frac{1}{m} \sum_{e\in E} \rel(e,v) 
	= \nrm{S_v}\,\alpha + \frac{1}{m} \Rel(v)
	\quad \text{for $v\in V$,}\\
\SC(S) &= m\,\alpha + \frac{1}{m} \sum_{v\in V} \Rel(v)\enspace.
\end{align*}
Clearly, this is a symmetric adversary
and hence disconnection cost is anonymous.
All results in this section are for the simple-minded adversary.
As promised earlier,
we give an example for a non-linear (in~$n$) social cost.
\begin{proposition}
\label{prop:sc-path}%
Social cost of a path is
$(n-1)\,\alpha + \frac{1}{3} n \, (n+1) = \Theta(n\alpha + n^2)$.
\end{proposition}
\begin{proof}
We have the social cost of a path:
\begin{align*}
\MoveEqLeft
m\,\alpha + \frac{1}{m} \sum_{v\in V} \sum_{e\in E} \rel(e,v)
 = (n-1)\,\alpha + \frac{1}{n-1} \sum_{e\in E} \sum_{v\in V} \rel(e,v)\\
& = (n-1)\,\alpha + \frac{1}{n-1} \sum_{e\in E} \sep(e) \\
& = (n-1)\,\alpha + \frac{1}{n-1} \sum_{e\in E} 2 \, \nu(e) \, (n-\nu(e))\\
& = (n-1)\,\alpha + \frac{2}{n-1} \sum_{k=1}^{n-1} k \, (n-k)\\
& = (n-1)\,\alpha + \frac{2}{n-1} \delim{(}{n \frac{(n-1)\,n}{2} - \frac{(n-1)\,n\,(2n-1)}{6}}{)}\\
& = (n-1)\,\alpha + \frac{1}{3} n \, (n+1) = \Theta(n\alpha + n^2)\enspace.
\qedhere
\end{align*}
\end{proof}

\subsection*{Bounding Cost Changes and Cycle Length}
We estimate the benefit for a player of building or selling
a particular edge.
This will become useful in several places.
It moreover immediately leads to a structural result 
on the length of cycles.
The following remark is purely graph-theoretic
and will be used here and later, in \autoref{sec:bilateral},
when we study convexity of cost.
\begin{remark}
Let $G=(V,E)$ be a graph.
\label{rem:one-cycle}%
\begin{enumerate}[label=(\roman*)]
\item\label{rem:one-cycle:i}%
Let $e=\set{v,w}\not\in E$
and $C$ be \emphasis{any} cycle in $G+e$ with $e \in E(C)$.
Then all bridges in $G$ that are non-bridges in $G+e$
are located on $C$.
\item\label{rem:one-cycle:ii}%
Let $e=\set{v,w}\in E$ be a non-bridge
and $C$ be \emphasis{any} cycle with $e\in E(C)$.
Then all bridges of $G-e$ that are non-bridges in $G$,
are in~$E(C)$.
\end{enumerate}
\end{remark}
\begin{proof}
\ref{rem:one-cycle:i}~The additional edge $e$ creates exactly one cycle $\tC$ in the bridge tree.
All bridges in $G$ that are non-bridges in $G+e$ correspond to edges on $\tC$,
and all those in turn correspond to edges on $C$.
\par
\ref{rem:one-cycle:ii}~Let $f$ be a non-bridge in $G$ and a bridge in $G-e$.
Then $G-e$ consists of two subgraphs $G_1$ and $G_2$ that are connected only by $f$.
Since $f$ was no bridge before $e$ was removed,
$e$ must also connect $G_1$ with $G_2$.
Moreover, there are no other edges between $G_1$ and $G_2$.
It follows that any cycle that contains $e$ also contains $f$.
\end{proof}
\begin{proposition}
For each player $v$ we have $\Rel(v)\leq \frac{n\,(n-1)}{2}$.
\end{proposition}
\begin{proof}
We repeat the counting argument from the proof of \autoref{lem:double-count}:
\begin{align*}
R(v) & \stackrel{\text{def}}{=}
\sum_{e\in E} \rel(e,v)
\stackrel{\text{\eqref{eqn:rel-paths}}}{=}
\sum_{\substack{e\in E\\\text{$e$ is a bridge}}} \card{\setst{P\in\shp(v)}{e\in E(P)}}\\
& = \sum_{P\in \shp(v)} \card{\setst{e\in E(P)}{\text{$e$ is a bridge}}}
\leq \sum_{P\in \shp(v)} \card{E(P)}
= \sum_{w\in V}\dist(v,w)
\enspace.
\end{align*}
This is maximal if $G$ is a path with $v$ at its end;
then $\Rel(v)=\frac{n\,(n-1)}{2}$.
\end{proof}
\par
Fix a player $v$.
Let $R\df\Rel(v)$ and let $R'$ be the same quantity when an additional edge~$e$ is built by~$v$.
By the previous proposition, we have $R,R'\leq \frac{n\,(n-1)}{2}$.
The benefit in disconnection cost of building this edge for player $v$ is
$\frac{1}{m} R
-\frac{1}{m+1} R'$.
Due to the change in denominators from \enquote{$m$} to \enquote{$m+1$}
this expression looks somewhat unhandy.
Yet, we can give good bounds incorporating the change in relevances,
\mbox{$\Delta R \df R - R' \geq 0$},
with one denominator.
We can do something similar 
for the case when the player \emphasis{sells} an edge,
where we put $\Delta R \df R' - R \geq 0$.
\begin{proposition}
\label{prop:delta-r}%
\hfill
\begin{enumerate}[label=(\roman*)]
\item\label{prop:delta-r:i}%
If a player builds an additional edge
and the sum of her relevances 
drops from $R$ to $R'$ by $\Delta R \df R - R'$,
then her improvement in disconnection cost is
at least $\frac{1}{m+1} \Delta R$
and at most
$\frac{1}{2} + \frac{1}{m+1} \Delta R \leq \frac{n}{2}$.
\item\label{prop:delta-r:ii}%
If a player sells a non-bridge and the sum of her relevances 
increases from $R$ to $R'$ by $\Delta R \df R' - R$,
then her impairment in disconnection cost is
at least $\frac{1}{m} \Delta R$
and at most
$\frac{1}{2} + \frac{1}{m} \Delta R \leq \frac{n}{2}$.
\end{enumerate}
\end{proposition}
\begin{proof}
\ref{prop:delta-r:i}~We have
\begin{align*}
\MoveEqLeft
\frac{1}{m} R - \frac{1}{m+1} R'
=\frac{1}{m} R - \frac{1}{m+1} \, \parens{R + (R'-R)}
=\parens{\frac{1}{m}-\frac{1}{m+1}} R + \frac{1}{m+1} \Delta R\\
&=\frac{1}{m\,(m+1)} R + \frac{1}{m+1} \Delta R
\quad
\begin{cases}
\leq \frac{1}{2} + \frac{1}{m+1} \Delta R \leq \frac{n}{2}\\
\geq \frac{1}{m+1} \Delta R
\end{cases}.
\end{align*}
We used $\Delta R \leq \frac{n\,(n-1)}{2}$ and $n-1\leq m$ 
for the upper bound.
\ref{prop:delta-r:ii} is proved alike,
using $n\leq m$ since the graph contains a cycle.
\end{proof}
\begin{proposition}
\label{prop:delta-r-cycle}%
\hfill
\begin{enumerate}[label=(\roman*)]
\item\label{prop:delta-r-cycle:iii}%
If a player builds an edge creating a cycle of length $\ell$,
the improvement in disconnection cost is at most
$\frac{1}{2} + \frac{1}{m+1} \, (\ell-1) \, (n-\frac{\ell}{2})$.
(The graph is allowed to already contain other cycles.)
\item\label{prop:delta-r-cycle:iv}%
If a player sells an edge destroying a cycle of length $\ell$,
the impairment in disconnection cost is at most
$\frac{1}{2} + \frac{1}{m} \, (\ell-1) \, (n-\frac{\ell}{2})$.
\end{enumerate}
\end{proposition}
\begin{proof}
\ref{prop:delta-r-cycle:iii}~Let $C=(v,e_1,v_1,\hdots,v_{\ell-1},e_{\ell},v)$ be any new cycle,
created by the new edge $e_{\ell}$ bought by~$v$.
By \autoref{rem:one-cycle}\ref{rem:one-cycle:i},
all edges for which a change in relevance occurs by adding $e_{\ell}$,
\ie all edges that were bridges and become non-bridges due to the new edge,
are located on this cycle.
In the best case, \ie in case of maximal improvement,
\begin{compactitemize}
  \item all $\ell$ edges were bridges before and became non-bridges now, and
  \item without the additional edge,
	$n-1$ vertices are reached from $v$ only through $e_1$,
	$n-2$ through the next edge, and so on;
	edge $e_{\ell-1}$ is relevant for $(n-(\ell-1))$ vertices.
\end{compactitemize}
It follows 
\begin{equation*}
\Delta R \leq \sum_{k=1}^{\ell-1} (n-k)= (\ell-1) \, n - \sum_{k=1}^{\ell-1} k
= (\ell-1) \, n - \frac{(\ell-1)\,\ell}{2} = (\ell-1) \, (n-\frac{\ell}{2})
\enspace.
\end{equation*}
The statement follows with \autoref{prop:delta-r}\ref{prop:delta-r:i}.
\par
\ref{prop:delta-r-cycle:iv}~By \autoref{rem:one-cycle}\ref{rem:one-cycle:ii},
we may consider any cycle that is destroyed.
The rest is the same calculation as for~\ref{prop:delta-r-cycle:iii}.
\end{proof}
\begin{proposition}
\label{prop:cycle-length}%
Let $\ell < \alpha + \frac{1}{2}$.
\begin{enumerate}[label=(\roman*)]
\item
\label{prop:cycle-length:i}%
If a player builds an edge
creating a cycle of length $\ell$,
she suffers an impairment in her cost.
\item
\label{prop:cycle-length:ii}%
If a player sells an edge
destroying a cycle of length $\ell$,
she experiences an improvement in her cost.
\end{enumerate}
\end{proposition}
\begin{proof}
\ref{prop:cycle-length:i}
By \autoref{prop:delta-r-cycle}\ref{prop:delta-r-cycle:iii}, 
the player suffers an impairment in her cost if 
\begin{equation*}
\alpha > \frac{1}{2} + \frac{1}{m+1} \parens{\ell-1} \parens{n-\frac{\ell}{2}}\enspace.
\end{equation*}
Since $m\geq n-1$, this is the case if
$\alpha > \frac{1}{2} + \frac{1}{n} \, (\ell-1) \, (n-\frac{\ell}{2})$,
which is the same as
$n \, (\alpha+\frac{1}{2})>\ell \, (n-\frac{\ell}{2} +\frac{1}{2})$.
Since $\ell\geq 3 \geq 1$, this is the case if
$n \, (\alpha+\frac{1}{2})>\ell n$.
\par
We show \ref{prop:cycle-length:ii} in almost exactly the same way,
using \autoref{prop:delta-r-cycle}\ref{prop:delta-r-cycle:iv}
and that $m\geq n$, since the original graph contains a cycle.
\end{proof}
It follows the structural result:
\begin{corollary}
\label{cor:no-short-cycles}%
No \NE\ contains cycles shorter than $\alpha+\frac{1}{2}$.
\hfill\qedsymbol
\end{corollary}

\subsection*{Bounding the Price of Anarchy}
The following observation
is the key to showing that a \NE\ does not have 
many more edges than a tree.
\begin{proposition}
\label{prop:chord-free}%
A \NE\ is chord-free.
\end{proposition}
\begin{proof}
Selling a chord $e=\set{v,w}$ from a cycle 
$C=(v,\hdots,w,\hdots,v)$ does not increase the relevance 
of any edge for any player.
To see this, we show that the bridge tree does not change by removal of $e$.
Assume for contradiction that there exists an edge $e'$
which is a bridge in $G' \df G - e$, but which is no bridge in $G$.
Then $G'-e'$ consists of two components $G_1$ and $G_2$.
Since $e'$ is no bridge in $G$, the edge $e$ connects $G_1$ and $G_2$.
But then, due to the existence of $C$,
in addition to $e$ there are \emphasis{two} more edges between $G_1$ and $G_2$.
Hence removal of the single edge $e'$ from $G'$ cannot disconnect $G_1$ from $G_2$.
\par
If the graph is bridgeless, removing a chord would thus decrease 
the player's building cost without increasing the disconnection cost.
Now let the graph contain a bridge $e'$.
Due to the decrease in the denominator of the disconnection cost,
removing a chord impairs the disconnection cost.
However, the player owning the chord, say~$v$, would rather remove the chord and 
instead build an edge to form a new cycle containing~$e'$.
The only case where this is impossible
is when $v$ is one endpoint of the bridge $e'=\set{v,u}$,
and $u$ is a leaf vertex.
Then, a double-edge between $v$ and $u$ would be needed,
which is not allowed unless we use a multigraph.
\par
We consider this case now
and show that we in fact do not need a multigraph.
By selling the chord, the disconnection cost for $v$ increases by $\frac{1}{m\,(m-1)}\Rel(v)$.
If this increase is strictly smaller than $\alpha$, we are done.
Hence assume $\frac{1}{m\,(m-1)}\Rel(v)\geq\alpha$ now.
Edge $\set{v,u}$ has relevance $n-1$ for $u$.
Due to the positions of $v$ and $u$, we have $\Rel(u)=\Rel(v)+(n-1)-1$.
If $u$ builds an edge to any other vertex, save $v$,
edge $e'$ is put on a cycle.
The improvement in disconnection cost for $u$ by building such an edge is at least
\begingroup\smaller
\begin{align*}
\MoveEqLeft
\frac{1}{m}\Rel(u) - \frac{1}{m+1} \, (\Rel(u)-(n-1))
  = \frac{1}{m} \, (\Rel(v)+n-2) - \frac{1}{m+1} \, (\Rel(v)-1) \\
& = \delim{(}{\frac{1}{m}-\frac{1}{m+1}}{)} \Rel(v)+ \frac{n-2}{m} + \frac{1}{m+1}
  = \frac{1}{m\,(m+1)} \Rel(v)+ \frac{n-2}{m} + \frac{1}{m+1} \\
& = \delim{(}{\frac{1}{m\,(m-1)}+\frac{1}{m\,(m+1)}-\frac{1}{m\,(m-1)}}{)} \Rel(v)+ \frac{n-2}{m} + \frac{1}{m+1} \\
& \geq \alpha - \frac{1}{m} \parens{\frac{1}{m-1}-\frac{1}{m+1}} \Rel(v)+ \frac{n-2}{m} + \frac{1}{m+1} \\
& \geq \alpha - \frac{1}{m} \parens{\frac{1}{m-1}-\frac{1}{m+1}} \frac{n\,(n-1)}{2} + \frac{n-2}{m} + \frac{1}{m+1} \\
& \geq \alpha - \parens{\frac{1}{m-1}-\frac{1}{m+1}} \frac{n-1}{2} + \frac{n-2}{m} + \frac{1}{m+1} \\
& = \alpha - \frac{2}{(m-1)\,(m+1)} \frac{n-1}{2} + \frac{n-2}{m} + \frac{1}{m+1} \\
& \geq \alpha - \frac{1}{m+1} + \frac{n-2}{m} + \frac{1}{m+1} > \alpha\enspace.
\end{align*}
\endgroup
So $u$ has an incentive to buy an additional edge,
a contradiction to \NE.
\end{proof}
\par
The next two are graph-theoretic results.
The first is a straightforward adaption of a result (and its proof) 
on vertex-connectivity to edge-connectivity;
see, \eg~\cite[\nref{Prop.}{3.1.3}]{Die05} for the version for vertex-connectivity.
\begin{proposition}
Any bridgeless connected graph can be constructed from a cycle 
by successively adding paths or cycles of the form $(u,e_1,v_1,\hdots,v_k,e_{k+1},w)$,
where $u,w$ are vertices of the already constructed graph ($u=w$ is allowed)
and $v_1,\hdots,v_k$ are zero or more new vertices.
\end{proposition}
\begin{proof}
Clearly, any graph that was constructed in this manner is connected and bridgeless.
Now let $G$ be connected and bridgeless
and $H$ a subgraph of $G$ that is constructible in this manner,
chosen such that it has a maximum number of edges among all such subgraphs.
Since $G$ contains a cycle, $H$ is not empty.
Also, $H$ is an induced subgraph since $H+e$ is also constructible for any edge $e$.
If $H\neq G$, then since $G$ is connected,
there is an edge $e=\set{v,w}$ with $v\not\in V(H)$ and $w\in V(H)$.
Since $G$ is bridgeless, this edge is on a cycle $C=(w,e,v=v_1,\hdots,v_k=w)$.
Let $v_i$ be the first vertex with $v_i\in V(H)$.
Then $P\df(w,\hdots,v_i)$ is a path or cycle of the form used in the construction,
and so $H+P$ is constructible and has more edges than $H$, 
a contradiction.
\end{proof}
\begin{proposition}
\label{prop:chord-free-general}%
A chord-free graph on $n$ vertices contains no more than $2n-1=O(n)$ edges.
\end{proposition}
\begin{proof}
Let $G$ be a chord-free graph, \wlg being connected.
We first consider the case that $G$ is bridgeless.
By the previous proposition, $G$ can be constructed 
from a cycle on, say, $N_0$ vertices,
by successively adding paths of the form $(u,e_1,v_1,\hdots,\linebreak v_k,e_{k+1},w)$,
where $u,w$ are vertices of the already constructed graph
and $v_1,\hdots,v_k$, $k\in\NNzero$, are zero or more new vertices.
For any two vertices $u,w$ in the already constructed graph,
there is a cycle $C$ with $u,w\in V(C)$.
Since $G$ is chord-free, we may not add a path $(u,e_1,w)$.
Hence $k\geq 1$ in each step,
\ie at least one new vertex is added.
It follows that there are at most $t\leq n-N_0\Df N_1 \leq n-1$ steps in this construction.
Let $n_i$ and $m_i$ be the number of new vertices and edges, respectively,
inserted in step $i$.
Then $m_i=n_i+1$ for each $i\in\setn{t}$ and so
we add $\sum_{i=1}^t m_i=\sum_{i=1}^t(n_i+1)=N_1+t\leq 2N_1$ edges
to the initial cycle.
It follows that $G$ has at most $N_0+2N_1 = n + N_1 \leq 2n - 1$ edges.
\par
If $G$ is not bridgeless, we consider its \bcc s $\eli{B}{r}$;
these correspond to the vertices of the bridge tree, and $r=\card{\tV}$ is the number of vertices of the bridge tree.
By what we proved above, the \bcc s contribute at most $\sum_{i=1}^r (2 \card{B_i} - 1) = 2n - r$ edges.
In addition, since the bridge tree is a tree, there are at most $r-1$ edges
(these are all the bridges of~$G$, or in other words all the edges of the bridge tree).
So we have a total bound of $2n-r+(r-1)=2n-1$.
\end{proof}
\begin{corollary}
\label{cor:number-edges-bound}%
A \NE\ has at most $2n-1=O(n)$ edges.
\end{corollary}
\begin{proof}
Follows from 
\autoref{prop:chord-free} and \autoref{prop:chord-free-general}.
\end{proof}
Now we know that the total building cost in a \NE\
is $O(n\alpha)$, hence it is of the same order as the optimal social cost.
In order to bound the price of anarchy,
we are left with bounding the disconnection cost.
To this end, we make use of the bridge tree.
The following is a corollary to \autoref{lem:double-count}.
\begin{corollary}
\label{cor:disconnect-diam-bound}%
The disconnection cost is bounded by $n\,\diam(\tG)$.
\end{corollary}
\begin{proof}
We have by \autoref{lem:double-count}:
\begin{align*}
\MoveEqLeft
\frac{1}{m} \sum_{v\in V} \sum_{e\in E} \rel(e,v)
= \frac{1}{m} \sum_{v\in V} R(v)
\leq \frac{1}{m} \sum_{v\in V} (n-1) \, \diam(\tG)\\
&= \frac{n}{m} \, (n-1) \, \diam(\tG)
\leq n \, \diam(\tG)\enspace.
\qedhere
\end{align*}
\end{proof}
A bound on the diameter of the bridge tree holding for all \NE\
will hence yield a bound on the price of anarchy.
This is accomplished by the following lemma.
\begin{lemma}
\label{lem:diam}%
The bridge tree of a \NE\ has its diameter bounded by $8 \al = O(\alpha)$.
\end{lemma}
\begin{proof}
Let $G$ be a \NE.
Let $\tP=(v_0,e_1,v_1,\hdots,e_\ell,v_\ell)$ 
be a path
in the bridge tree $\tG$ connecting two leaves $v_0$ and $v_\ell$.
Let $\bar{\ell}\df\sceil{\frac{\ell}{2}}\geq1$.
Then at least one of the following is true 
(recall the convention \vpageref{con:vertex-counting}
regarding vertex-counting in the bridge tree):
\begin{compactitemize}
\item At least $\sceil{\frac{n}{2}}$ vertices lie beyond $e_{\bar{\ell}}$
from the view of $v_0$.
\item At least $\sceil{\frac{n}{2}}$ vertices lie beyond $e_{\bar{\ell}}$
from the view of $v_\ell$.
\end{compactitemize}
\par
Let us assume the first; the other case can be treated alike.
Let $v\df v_0$ and $w\df v_{\ell}$
and recall that we may treat vertices of the bridge tree $\tG$ as single players with respect to building of new links.
Then $e_1,\hdots,e_{\bar{\ell}}$ 
for $v$ have relevance at least $\sceil{\frac{n}{2}}$ each.
So $\sum_{i=1}^{\bar{\ell}}\rel(e_i,v)\geq \bar{\ell} \, \frac{n}{2}\geq \frac{\ell}{2} \, \frac{n}{2} =\Omega(\ell n)$.
By building $\set{v,w}$, player $v$ would have a benefit in disconnection cost
of at least $\frac{1}{m+1} \frac{\ell n}{4} \geq \frac{1}{2n} \frac{\ell n}{4} 
=\Omega(\ell)$,
using the bound $m\leq 2n-1$ from \autoref{cor:number-edges-bound}.
Since the edge is not built, $\alpha$ is larger than this benefit, 
so $\ell \leq 8 \al = O(\alpha)$.
\end{proof}
\begin{corollary}
\label{cor:disconnection-diam-bound-alpha}%
The disconnection cost in a \NE\ is bounded by $8n\al = O(n\alpha)$.
\end{corollary}
\begin{proof}
Follows from \autoref{cor:disconnect-diam-bound}
and \autoref{lem:diam}.
\end{proof}
\begin{theorem}
\label{thm:uni-simple-O1}%
The price of anarchy 
with a simple-minded adversary is bounded by $O(1)$.
\end{theorem}
\begin{proof}
The building cost and the disconnection cost in a \NE\ are both $O(n\alpha)$
by \autoref{cor:number-edges-bound} and \autoref{cor:disconnection-diam-bound-alpha}.
The theorem follows with \autoref{prop:uni-opt},
which states that the optimum social cost is $\Theta(n\alpha)$.
\end{proof}
\par
A closer look at \autoref{lem:diam} and its proof reveals that
there exists a constant $c>0$ such that if $m=O(n)$
then there are players who can improve their disconnection cost by~$c$
through the building of new links,
as long as the graph contains bridges.
It follows that for $\al<c$, all \NE\ are bridgeless,
\ie they have disconnection cost $0$
and the adversary cannot harm them.
This does not rule out, however,
that they may contain an unnecessarily high number of links,
compared to an optimum.
On the other hand, the ratio cannot be more than $O(1)$
by \autoref{thm:uni-simple-O1}.
\begin{remark}
The constant in \autoref{thm:uni-simple-O1} is bounded by $10 + \frac{10}{n-1}=10+o(1)$.
\end{remark}
\begin{proof}
Building cost of a \NE\ is bounded by $2n\al$
by \autoref{cor:number-edges-bound}.
Disconnection cost of a \NE\ is bounded by $8 n \al$ 
by \autoref{cor:disconnection-diam-bound-alpha}.
In total, social cost of a \NE\ is bounded by $10n \al$.
Using \autoref{rem:concrete-bound-general} with $c_1 \df 10$ and $c_0 \df 0$
proves the claim.
\end{proof}

\section{Smart Adversary}
\label{sec:uni-smart}%
We remain with \ULF\ and consider an adversary that destroys an edge
which separates a maximum number of vertex pairs.
If there are several such edges, one of them is chosen uniformly at random.
In other words, we replace the uniform probability distribution on the edges
for one that is concentrated on the edges which cause maximum overall damage.
Recall that $\sep(e)$ is
the number of separated vertex pairs when edge $e$ is deleted.
Let $\sep_{\max}\df \max_{e\in E}\sep(e)$
and $\Emax\df \setst{e\in E}{\sep(e)=\sep_{\max}}$
and $\mmax\df \card{\Emax}$.
These are the edges of which each causes a maximum number of
separated vertex pairs when it is deleted.
We call them the \term{critical} edges.
The adversary chooses one of those uniformly at random.
Clearly, this yields a symmetric adversary,
and so disconnection cost is anonymous.
We have the individual and social cost:
\begin{align*}
C_v(S) &= \nrm{S_v}\,\alpha + \frac{1}{\mmax} \sum_{e\in \Emax} \rel(e,v) & \text{for $v\in V$,}\\
\SC(S)
&= m\,\alpha + \frac{1}{\mmax} \sum_{e\in \Emax} \sep_{\max}
= m\,\alpha + \sep_{\max}\enspace.
\end{align*}
\par
If $\sep_{\max}=0$, then the graph is bridgeless
and all edges are critical 
-- however, their removal does not separate any vertex pairs.
If \mbox{$\sep_{\max}>0$}, then there are one or more critical edges,
and each of them is a bridge.
Recall that if $e$ is a bridge,
$\nn(e)$ denotes the number of vertices in the smaller component of $G-e$,
or $\frac{n}{2}$ if both are of equal size.
If $e$ is no bridge, then $\nn(e)=0$.
We have $\sep(e)=2\nn(e)\,(n-\nn(e))$ for all edges.
The function $x\mapsto2x\,(n-x)$ is strictly increasing 
on $[0,\frac{n}{2}]$,
so $\nn(e)=\nn(e')$ follows from $\sep(e)=\sep(e')$.
Hence $\nn(e)=\nn(e')$ for all critical edges $e,e'\in\Emax$.
\begin{proposition}
\label{prop:critical-star}%
If $\sep_{\max}>0$ and if there are more than one critical edges,
they form a subgraph that is a star in the bridge tree $\tG$.
\end{proposition}
\begin{proof}
Let $\sep_{\max}>0$.
For any two distinct bridges $e$ and $e'$,
one component of $G-e$ is strictly contained in one component of $G-e'$.
Therefore, with multiple critical edges, $\nn(e)<\frac{n}{2}$ for all $e\in E_{\max}$,
and so also for all other bridges (since they have smaller $\nn(\cdot)$ value).
In other words, there is always a small and a large component of $G-e$,
with $e$ being a bridge.
\par
Let $P=(v_0,e_1,v_1,\hdots,v_{\ell-1},e_{\ell},v_{\ell})$
be a path in the bridge tree $\tG$ with $e_1$ and $e_{\ell}$ being distinct critical edges.
First assume that $v_{\ell}$ is in the larger component of $G-e_{\ell}$.
Then $v_0$ is in the smaller component of $G-e_1$.
Then the smaller component of $G-e_2$ cannot contain $v_0$, 
since otherwise $\nn(e_1)<\nn(e_2)$,
and $e_1$ would not be critical.
So the component of $G-e_2$ containing $v_0$ is the larger one,
and then the same holds for the component of $G-e_{\ell}$ containing $v_0$.
This contradicts that $v_{\ell}$ is in the larger component of $G-e_{\ell}$.
We can carry out the same argument with $v_0$ and $e_1$.
Summarizing, now we know that the smaller component of $G-e_1$ is located `before' $P$
and that the smaller component of $G-e_{\ell}$ is located `beyond'~$P$.
\par
If $\ell\geq 3$, 
then there is an edge $f$ between $e_1$ and $e_{\ell}$ on $P$.
The smaller component of $G-f$ strictly contains 
either the smaller component of $G-e_1$ or $G-e_{\ell}$.
Since $\nn(e_1)=\nn(e_2)$,
we have thus in particular, $\nn(f)>\nn(e_1)$, 
a contradiction that $e_1$ is critical.
Hence there is no such edge $f$, and so $\ell=2$.
Since this holds for all pairs $(e_1,e_{\ell})$ of critical edges,
the set of all critical edges forms a star (in the bridge tree).
\end{proof}
\begin{figure}
\newcommand{\skel}{%
	\node (c) { };
	\node (r1) [right = of c]  { };
	\node (r2) [right = of r1] { };
	\node (r3) [right = of r2] { };
	\node (r4) [right = of r3] { };
	\node (l1) [left = of c]  { };
	\node (l2) [left = of l1] { };
	\node (l3) [left = of l2] { };
	\node (l4) [left = of l3] { };
	\path (r1) edge (r2)
		  (r2) edge (r3)
		  (r3) edge (r4);
	\path (l1) edge (l2)
		  (l2) edge (l3)
		  (l3) edge (l4);}
\centering
\subfloat[]%
[\label{fig:uni-ne-path:1}A path with all edges pointing to the nearest end.]
{\begin{tikzpicture}[%
	every node/.style=player,%
	every edge/.style=dirlink,%
	node distance=1.65cm]
	\skel
	\path (c) edge[dashed] (r1);
	\path (c) edge[dashed] (l1);
\end{tikzpicture}}
\vspace{.5cm}
\subfloat[]%
[\label{fig:uni-ne-path:1}Center vertex cannot improve by exchanging her edges.]
{\begin{tikzpicture}[%
	every node/.style=player,%
	every edge/.style=dirlink,%
	node distance=1.65cm,%
	bend angle=40]
	\skel
	\path (c) edge[dashed,bend right] (r2);
	\path (c) edge[dashed,bend left] (l3);
\end{tikzpicture}}
\caption{%
	\label{fig:uni-ne-path}%
	\NE\ if $\alpha\geq\frac{n}{2}$.
	Critical edges are drawn dashed.
	Disconnection cost for the center vertex is $\sfloor{\frac{n}{2}}$ in both cases,
	which is $4$ here.}
\end{figure}
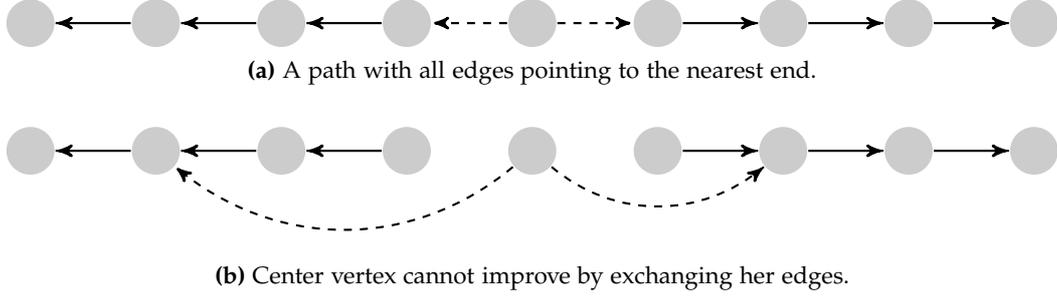%
The smart adversary admits a new \NE\ topology:
\begin{proposition}
\label{prop:uni-ne-path}%
If $\alpha \geq \frac{n}{2}$,
then a path with all edges pointing to the nearest end
(in case of even $n$, the middle edge having arbitrary orientation)
is a \NE\ with social cost $\Theta(n\alpha+n^2)$.
If $\alpha > \frac{n}{2}$, it is a \MNE.
\end{proposition}
\begin{proof}
The social cost of the path is
$(n-1)\,\alpha + \sfloor{\frac{n}{2}}\sceil{\frac{n}{2}}=\Theta(n\alpha+n^2)$.
The adversary removes the one or two
-- depending on whether $n$ is even or odd -- middle edges.
The disconnection cost for each player is $\frac{n}{2}\leq\alpha$
if $n$ is even
and at most
$\frac{1}{2}\,(\sfloor{\frac{n}{2}}+\sceil{\frac{n}{2}})=\frac{n}{2}\leq \alpha$ if $n$ is odd.
Hence, there is no incentive for any player to build more edges than she currently owns, 
even after exchanging the currently built edges for others.
\par
Now consider that a player $v$ sells one (or two) of her edges 
and buys one (or two) different ones instead.
First consider that one edge is exchanged.
Since all edges point outward,
the part of the path that becomes disconnected from $v$ 
does not contain the critical edge(s).
So, after reconnecting it with~$v$ via a new edge,
there are as many vertices on both sides of the formerly critical edge(s)
as before the exchange.
No separation value increases.
Hence the formerly critical edges remain critical.
They also maintain their relevance for $v$.
With the same argument,
if the exchanged edge itself was critical,
the new one will be critical as well, 
also with the same relevance for $v$.
\par
When two edges are exchanged,
$v$ is the center vertex,
and in particular all critical edge(s) are among the exchanged ones,
see \autoref{fig:uni-ne-path}.
This again means that the disconnected parts
do not contain critical edges,
and so the exchange cannot 
change that each of the two edges has $\sfloor{\frac{n}{2}}$ vertices on
the one and $\sceil{\frac{n}{2}}$ vertices on the other side,
so they remain critical.
Also, their relevance for $v$ does not change.
The \MNE\ property is clear by \autoref{rem:mne}.
\end{proof}

\subsection*{Bounding the Price of Anarchy}
The proof of the following is even easier than previously:
\begin{remark}
\label{rem:chord-free-smart}%
A \NE\ is chord-free.
\end{remark}
\begin{proof}
Removing a chord does not change the relevance of any edge,
nor does it change $\sep_{\max}$,
hence it does not change $\Emax$.
Selling a chord so is always beneficial.
\end{proof}
\par
With \autoref{prop:chord-free-general}, it follows immediately:
\begin{corollary}
\label{cor:bound-edges-smart}%
A \NE\ has at most $2n-1=O(n)$ edges.
\hfill\qedsymbol
\end{corollary}
\par
We are again left with bounding the disconnection cost of \NE.
This requires some effort and is accomplished in the following 
remark and two lemmas.
\begin{remark}
\label{rem:critical-remains}%
If there are $k\geq 2$ critical edges, say $\Emax=\set{e_1,\hdots,e_k}$,
and $e_1$ is put on a cycle by an additional edge, but not $e_2,\hdots,e_k$,
then the new critical edges are $e_2,\hdots,e_k$. 
If $k\geq 3$ and the additional edge puts $e_1$ and $e_2$ on a cycle,
but not $e_3,\hdots,e_k$,
then the new critical edges are $e_3,\hdots,e_k$.
\end{remark}
\begin{proof}
An additional edge $e$ only changes the $\nn(\cdot)$ value
of those edges which are put on a cycle by $e$,
namely it reduces them to $0$.
Hence, none of the edges in $\set{e_2,\hdots,e_k}$ 
(or $\set{e_3,\hdots,e_k}$) becomes less attractive for the adversary 
when $e$ is added.
Also no other edge becomes more attractive by the addition of~$e$,
since no $\nn(\cdot)$ value increases.
\end{proof}
\begin{lemma}
\label{lem:mmax-3}%
Let $\al \leq c n$ for a constant $c>0$ and fix a \NE\ with $\mmax \geq 3$.
Then we have $\sep_{\max} \leq 2 \, (1+9c) \, n\al = O(n \al)$.
\end{lemma}
\begin{proof}
Fix two critical edges $e_1$ and $e_2$,
and set $n_0\df\nn(e_1)$.
For each $i\in\set{1,2}$ fix a player $v_i$ 
in the smaller component of $G-e_i$.
Then for each $i\in\set{1,2}$ we have $\rel(e_i,v_i)=n-n_0$
and $\rel(e,v_i)=n_0$ for all critical edges $e\neq e_i$;
recall that all critical edges have the same $\nn(\cdot)$ value.
Building $\set{v_1,v_2}$ puts $e_1$ and $e_2$ on a cycle
and leaves the other $\mmax-2$ critical edges critical 
by \autoref{rem:critical-remains}.
For each $i\in\set{1,2}$,\footnote{%
It would suffice to restrict to $i=1$ or $i=2$.
However, here and in the proof of the following \autoref{lem:mmax-2-1},
we point out all arguments that are 
symmetric in the sense that both endpoints would like to build the edge.
This is interesting for \BLF\ discussed in \autoref{sec:bilateral}.}
player $v_i$ has her disconnection cost decreased by:
\begin{align*}
\MoveEqLeft
\frac{1}{\mmax} \sum_{e\in \Emax} \rel(e,v_i) 
- \frac{1}{\mmax-2} \sum_{\substack{e\in \Emax\\e\not\in\set{e_1,e_2}}} \rel(e,v_i) \\
&= \frac{1}{\mmax} \, ( (\mmax - 1) \, n_0 + n - n_0 )
- \frac{1}{\mmax-2} \, (\mmax - 2) \, n_0 \\
&= \frac{1}{\mmax} \, ( (\mmax - 2) \, n_0 + n )
- n_0 = \frac{1}{\mmax} \, ( n-2n_0) \enspace.
\end{align*}
Since we are in a \NE, this is at most $\alpha$.
Since $n\geq \mmax n_0$,
we have $n-2n_0 \geq (\mmax-2) \, n_0$,
and so it follows
$\alpha \geq (1-\frac{2}{\mmax}) \, n_0 \geq \frac{1}{3} n_0$.
Moreover, it follows
$\mmax\alpha + n_0 \geq n - n_0$.
With these two inequalities at hand, we can bound $\sep_{\max}$.
We have
\begin{align*}
\sep_{\max} & = 2 n_0 \, (n-n_0)
\leq 2 n_0 \, (\mmax\alpha + n_0)
\leq 2 \, (n_0 \mmax \alpha + 9 \alpha^2) \\
& \leq  2 \, (n \alpha + 9 \alpha^2)
\stackrel{\al \leq cn}{\leq}  2 \, (n \alpha + 9 cn \alpha)
=  2 \, (1 + 9 c) \, n \al
\enspace.\qedhere
\end{align*}
\end{proof}
\begin{lemma}
\label{lem:mmax-2-1}%
Fix a \NE\ with $\mmax \in \set{1,2}$.
Then
\begin{compactitemize}
\item we have $\sep_{\max} \leq 4 n \al = O(n \al)$
\item or we have $\al \geq \frac{1}{6} n = \Omega(n)$.
\end{compactitemize}
\end{lemma}%
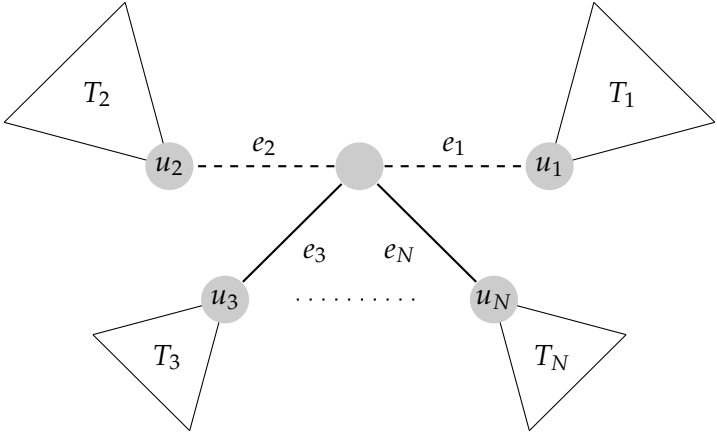
\begin{SCfigure}
\caption{%
	\label{fig:tree-structure}%
	Schematic view of the bridge tree with two critical edges $e_1$ and $e_2$, drawn dashed.
	Subtrees are represented by triangles.}
\begin{tikzpicture}[scale=.9]
	\TreeStructureSkel
	\path (c) edge[dashed] node {$e_1$} (u1);
	\path (c) edge[dashed] node[swap] {$e_2$} (u2);
	\path (c) edge node {$e_3$} (u3);
	\subtree{u2}{135}{2.5}{$T_2$}
\end{tikzpicture}
\end{SCfigure}%
\begin{SCfigure}
\caption{%
	\label{fig:insert-edge-structure}%
	How $e_1$ and $e_2$ are put on a cycle 
	by a new edge $\{v_1,v_2\}$.
	Paths that are part of the new cycle and located inside $T_1$ and $T_2$ 
	are depicted as zig-zag paths.
	New critical edges can emerge, \eg $e_3$ can become critical.}
\begin{tikzpicture}[scale=.9]
	\TreeStructureSkel
	\path (c) edge node {$e_1$} (u1);
	\path (c) edge node[swap] {$e_2$} (u2);
	\path (c) edge[dashed] node {$e_3$} (u3);
	\subtree{u2}{135}{2.5}{$T_2$}
	\node[player] (v1) [above right = 1.7 and 1 of u1] {$v_1$};
	\node[player] (v2) [above left = 1.7 and 1 of u2] {$v_2$};
	\path (v1) edge[decorate,decoration=zigzag] (u1);
	\path (v2) edge[decorate,decoration=zigzag] (u2);
	\path (v1) edge (v2);
\end{tikzpicture}
\end{SCfigure}
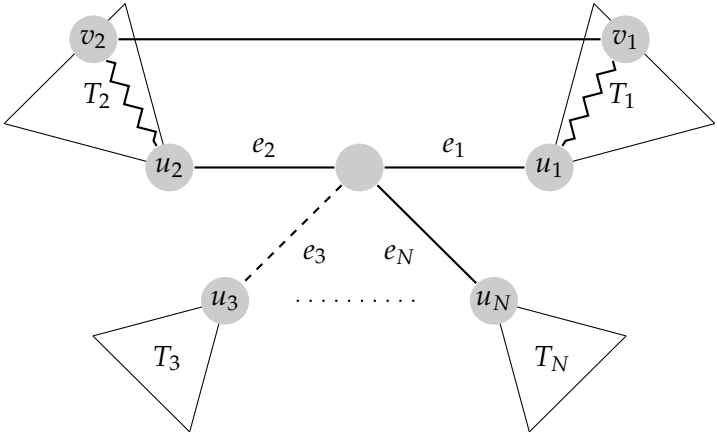%
\vspace{-.7\baselineskip}
\begin{proof}
First we consider the case $\mmax=2$.
A player can make the two critical edges part of a cycle 
by building an additional edge.
The difficulty lies in that new critical edges,
with a smaller separation value, can emerge.
We will have to put some more effort into estimating 
the improvement in disconnection cost that a player is able to achieve by building another edge.
Consider the bridge tree.
There are two subtrees $T_1$ and $T_2$ 
that are connected to the rest by the two critical edges $e_1$ and $e_2$, respectively.
They both have $n_0\df\nn(e_1)=\nn(e_2)$ vertices.
There may be more subtrees $T_3,\hdots,T_N$ connected by $e_3,\hdots,e_N$ to the center vertex.
To streamline notation, we often write $T_k$ instead of $V(T_k)$, $k\in\setn{N}$,
when we refer to the set of vertices of a tree.
\autoref{fig:tree-structure} 
\vpageref{fig:tree-structure} 
depicts the situation.
\autoref{fig:insert-edge-structure}
shows how a new edge would put $e_1$ and $e_2$ on a  cycle.
\par
First assume that we can arrange $v_1\in T_1$ and $v_2\in T_2$ 
such that after building $\set{v_1,v_2}$,
there are no critical edges in $T_1$ nor in $T_2$.
If there are no subtrees except $T_1$ and $T_2$, \ie if $N=2$,
this means that we can make the graph bridgeless by the additional edge.
The improvement in disconnection cost for $v_1$ (and also for $v_2$) of building $\set{v_1,v_2}$ is hence 
their original disconnection cost, \ie
$\frac{1}{2} \, (n-n_0 + n_0) = \frac{1}{2} n \geq \frac{1}{6} n$.
If $N\geq 3$, then critical edges emerge in one or more of the $T_3+e_3,\hdots,T_N+e_N$ after building.
Fix $k\in\set{3,\hdots,N}$.
Since $e_k$ is not critical without the new edge,
we have $\card{T_k}<n_0$ or $\card{T_k}\geq\sceil{\frac{n}{2}}$.
The latter can be excluded,
since it would imply that the smaller (or equally sized) component of $G-e_k$
includes $T_1$ and the center vertex, and so $\nn(e_k)>n_0=\nn(e_1)$,
in which case $e_1$ would not be critical.
Moreover, we have $\card{T_k}\leq n-2n_0 < n-2\,\card{T_k}$,
so $\card{T_k}<\frac{1}{3} n$.
For a player in $T_1$ (or $T_2$),
a critical edge in $T_k+e_k$ can have relevance at most $\card{T_k}$ 
and so no more than $\frac{1}{3} n$.
The improvement in disconnection cost for $v_1$ (and also for $v_2$) gained by building $\set{v_1,v_2}$ is hence at least
the original disconnection cost minus $\frac{1}{3}n$, \ie
$\frac{1}{2} \, (n-n_0 + n_0) - \frac{1}{3} n = \parens{\frac{1}{2} - \frac{1}{3}} n = \frac{1}{6} n$.
\par
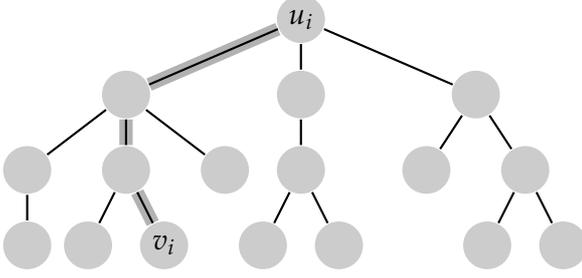
\begin{SCfigure}
\caption{%
	\label{fig:insert-edge-tree}%
	Detailed view of $T_i$ for one \mbox{$i\in\set{1,2}$}.
	Path $P_i$ is highlighted.
	Recall that vertices of the bridge tree are counted
	according to the size of the respective \bcc s.
	Here, in this example, we assume that each vertex counts~$1$.
	The path is drawn accordingly, \ie 
	always descending into a subtree with a maximum number of vertices.}
\begin{tikzpicture}[%
	every node/.style=player,%
	level distance=1cm,%
	level 1/.style={sibling distance=2.3cm},%
	level 2/.style={sibling distance=1.3cm},%
	level 3/.style={sibling distance=1.0cm},%
	]
	\node (u) {$u_i$} 
		child {node (x1) { } 
			child {node { } 
				child {node { }}}
			child {node (x2) { }
				child {node { }}
				child {node (v) {$v_i$}}}
			child {node { }}}
		child {node { }
			child {node { }
				child {node { }}
				child {node { }}}}
		child {node { }
			child {node { }}
			child {node { }
				child {node { }}
				child {node { }}}};
	\begin{pgfonlayer}{background}
	\path (u) edge[highlight] (x1)
		(x1) edge[highlight] (x2)
		(x2) edge[highlight] (v);
	\end{pgfonlayer}
\end{tikzpicture}
\end{SCfigure}%
Now consider that for all choices of $v_1\in T_1$ and $v_2\in T_2$,
building $\set{v_1,v_2}$ induces a critical edge in at least one of $T_1$ or $T_2$.
For each $i\in\set{1,2}$ we can do the following.
Let $u_i$ be the vertex where $T_i$ is connected to the rest of the graph
and consider $T_i$ being rooted at~$u_i$.
Let $P_i$ be a path starting at $u_i$ 
and ending at one of the leaves of $T_i$, say $w_i$,
such that the path always descends 
into a subtree that has a maximum number of vertices,
as shown in \autoref{fig:insert-edge-tree}.
If we choose $v_i\df w_i$, $i=1,2$,
then each $P_i$ does not contain a critical edge when we build $\set{v_1,v_2}$,
since these paths then both are located on a cycle.
However, by assumption, there is a critical edge $f$ in, say $T_1$.
By construction of $P_1$, we have $\nn(f)\leq \frac{n_0}{2}$.
So, player $v_1$ (and also $v_2$) can reduce her disconnection cost to 
no more than~$\frac{n_0}{2}$.
It follows that the improvement in disconnection cost is at least
$\frac{1}{2} \, (n-n_0 + n_0) - \frac{1}{2} n_0 = \frac{1}{2} \, (n-n_0)$,
which is at most $\alpha$, since we are in a \NE.
It follows
$\sep_{\max} = 2 n_0 \, (n-n_0) \leq 2 n_0 \cdot 2 \alpha \leq 4 n\alpha$.
\par\smallskip
The case of $\mmax=1$ can be treated similarly.
Let $e_1$ be the critical edge
and $T_1$ the subtree with $n_0\df\nn(e_1)$ vertices.
There are zero or more additional subtrees, say $T_2,\hdots,T_N$.
If there are zero such trees,
define $T_2\df\tG-T_1$, which consists of just one vertex in the bridge tree then
(but can consist of multiple vertices in $G$).
Let the ordering be such that $\card{T_2} \geq \card{T_k}$ 
for all $k\in\set{3,\hdots,N}$.
Then we argue similar to before with $T_1$ and $T_2$
in the roles of the former subtrees of the same name.
Assume first that we can
find $v_1\in T_1$ and $v_2\in T_2$
such that building $\set{v_1,v_2}$ does not induce any critical edges
in $T_1$ nor~$T_2$.
If $N\leq 2$, then we can make the graph bridgeless
and this means an improvement for $v_1$ of at least 
$n-n_0$, and so 
$\sep_{\max} = 2 n_0 \, (n-n_0) \leq 2 n_0 \, \alpha \leq 4 n\alpha$.
If $N\geq 3$,
then fix $k\in\set{3,\hdots,N}$.
We have $\card{T_k}<n_0$.
Moreover, we have $\card{T_k} \leq n - (\card{T_2}+n_0) \leq n - 2 \, \card{T_k}$,
and so $\card{T_k} \leq \frac{1}{3} n$.
Then building $\set{v_1,v_2}$ 
reduces the disconnection cost of $v_1$ to no more than $\frac{1}{3}n$.
This means an improvement for $v_1$ of at least 
$n-n_0-\frac{1}{3}n\geq \frac{2}{3} n - \frac{1}{2}n =\frac{1}{6} n$.
\par
If each choice of $v_1$ and $v_2$
induces a critical edge in $T_1$ or $T_2$,
we can, as before, show that by a careful choice of these vertices,
building $\set{v_1,v_2}$ reduces the disconnection cost for $v_1$ (and $v_2$)
to at most~$\frac{n_0}{2}$.
Player $v_1$  originally has disconnection cost $n-n_0\geq n_0$,
so she experiences an improvement of at least~$\frac{n_0}{2}$.
(Player $v_2$ originally has disconnection cost $n_0$, 
so she as well experiences an improvement of at least~$\frac{n_0}{2}$.)
It follows $n_0\leq 2\alpha$ and so 
$\sep_{\max} = 2 n_0 \, (n-n_0) \leq 4 \alpha n$.
\end{proof}
\begin{theorem}
\label{thm:smart-bound}%
The price of anarchy with a smart adversary is $O(1)$.
\end{theorem}
\begin{proof}
Let $c > 0$ be the constant from the \enquote{$\al = \Om(n)$} statement of \autoref{lem:mmax-2-1},
\eg we may choose $c \df \frac{1}{6}$.
Consider first $\al < cn$.
We use $c$ as the constant in the premise in \autoref{lem:mmax-3}.
So if $\mmax \geq 3$, then \autoref{lem:mmax-3} gives $\sep_{\max}=O(n\al)$.
Otherwise, if $\mmax \in \set{1,2}$, \autoref{lem:mmax-2-1} gives the same,
since $\al \geq cn$ is ruled out.
Since $\sep_{\max}$ is the total disconnection cost,
it so has a ratio of $O(1)$ to the optimum.
\autoref{cor:bound-edges-smart} ensures that the same holds for the building cost.
\par
If $\al \geq cn$,
then this and \autoref{cor:bound-edges-smart}
allow us to invoke \autoref{cor:simple-bound}\ref{cor:simple-bound:iii}.
\end{proof}
\begin{remark}
The constant in \autoref{thm:smart-bound} is bounded by $8 + \frac{8}{n-1} = 8 + o(1)$.
\end{remark}
\begin{proof}
We proceed as in the proof of the theorem,
but do more detailed calculations.
Let $c \df \frac{1}{6}$.
We start again with $\al < cn$.
If $\mmax \geq 3$, 
then \autoref{lem:mmax-3} gives 
$\sep_{\max} \leq 2 \, (1+9c) \, n\al = 5n\al$.
Otherwise, if $\mmax \in \set{1,2}$, \autoref{lem:mmax-2-1} gives 
$\sep_{\max} \leq 4n\al$.
\par
Next consider $\al \geq cn$, \ie $\frac{\al}{c} \geq n$.
Using the trivial bound $n^2$ on the disconnection cost 
yields the $\frac{\al}{c} n = 6 n \al$ bound on it.
\par
In all cases disconnection cost is bounded by $6 n \al$.
Building cost is bounded by $2 n \al$ by \autoref{cor:bound-edges-smart}.
Finally invoking \autoref{rem:concrete-bound-general} with $c_1 \df 8$ and $c_0 \df 0$
yields a bound on the price of anarchy of $8 + \frac{8}{n-1}$.
\end{proof}

\section{Bilateral Link Formation}
\label{sec:bilateral}%
Recall that bilateral link formation (\BLF) means that $v$ requires $w$'s consent
to build $\set{v,w}$,
and if they both agree, they pay $\al$ each.
This is expressed by the bilateral final graph $G^\Bi=(V,E^\Bi(S))$ with
\begin{equation*}
E^\Bi(S)\df\setst{\set{v,w}}{S_{vw}=1\land S_{wv}=1}\enspace.
\end{equation*}
Social cost is $\SC(S) = 2\,\card{E(S)}\,\al + \sum_{v\in V}I_v(G(S))$,
differing from the social cost under \ULF\
only in a factor $2$ in building cost.
Appropriate equilibrium concepts are \PNE\ and \PS.
In general, the former implies the latter.
The converse holds if cost is convex.
\par
Restricting to essential strategy profiles,
individual and social cost are both determined by the final graph,
as are the properties of being a \PNE\ or being \PS.
Hence we will sometimes work with the final graph
in place of a strategy profile.
\par
Certain classes of simple-structured \MNE\ under \ULF\
are \PNE\ under \BLF.
The following two propositions hold in a more general setting,
not limited to the adversary model.
Let $S$ be a strategy profile.
Define $S^\Bi$ by
$S^\Bi_{vw} \df \min\set{1, S_{vw}\!+\!S_{wv}}$ for all $v,w\in V$.
Then $G^\Uni(S)=G^\Bi(S^\Bi)$.
In other words, forming $S^\Bi$ means adding to $S$
the necessary requests so that for \BLF\ the same final graph 
emerges as we have for \ULF.
\begin{proposition}
Let $S$ be a \MNE\ under \ULF\
with $G\df G^\Uni(S)$ being a cycle
and using anonymous indirect cost.
Then $S^\Bi$ is a \PNE\ under \BLF.
\end{proposition}
\begin{proof}
By the definition of \MNE,
any additional link is an impairment for the buyer.
So the premise of~\eqref{eqn:pne} \vpageref{eqn:pne} is never true,
\ie all absent edges are justified.
\par
New edges cannot be formed unilaterally.
We are hence left to show that each edge is wanted by both endpoints,
\ie none of the endpoints can improve her individual cost by deleting the edge.
Let $v$ be the owner of an edge $\set{v,w}$ under \ULF.
Since we have a \NE\ there, 
$v$ cannot improve her individual cost by selling this edge.
Selling the edge means that $v$ would be at the end of the path $G-\set{v,w}$.
By anonymity of indirect cost we conclude:
it is worth or at least no impairment paying $\alpha$ for not being at the end of the path
that results from $G$ by deletion of one edge.
Therefore, both of each two neighboring vertices 
maintain their requests in $S^\Bi$ for having an edge between them.
\end{proof}
\begin{proposition}
Let $S$ be a \MNE\ under \ULF\
with $G\df G^\Uni(S)$ being a tree.
Let the indirect cost assign $\infty$ to a disconnected graph.
Then $S^\Bi$ is a \PNE\ under \BLF.
\end{proposition}
\begin{proof}
As in the previous proposition,
\eqref{eqn:pne} follows from the properties
of a \MNE.
So we are left to consider removals.
Since the final graph is a tree,
removal of any edge would make it disconnected, yielding indirect cost $\infty$.
Hence no player wishes to remove an edge.
\end{proof}
Now we turn to the adversary model.
It follows from the two previous propositions that the equilibrium existence results from
\autoref{prop:uni-ne-star}, \autoref{prop:uni-ne-cycle},
and \autoref{prop:uni-ne-path} carry over from \MNE\ to \PNE.
Hence for anonymous disconnection cost,
we have proved existence of \PNE, and hence also \PS\ graphs,
under \BLF\ for all ranges of $\alpha$ and $n\geq 9$.
\par
\autoref{cor:simple-bound} holds independently of the equilibrium concept.
The cycle is optimal for $\alpha\leq n-1$
with social cost $2n\alpha$,
and a star is optimal for $\alpha\geq n-1$
with social cost $2\,(n-1)\,(\al+1)$;
so the optimum social cost is always $\Theta(n\alpha)$.
Ranges for $\alpha$ and the exact expressions for the social cost
are different from those in \autoref{prop:uni-opt},
accounting for the factor $2$ in building cost.
Otherwise, arguments are the same.
It can be checked easily that \autoref{thm:uni-stability} also carries over,
to \PNE\ as well as \PS.

\subsection*{Simple-Minded Adversary}
We show that for the simple-minded adversary,
cost is convex,
hence \PNE\ and \PS\ are equivalent.
The following proposition is purely graph-theoretical.
\begin{proposition}
\label{prop:becoming-bridges}%
Let $G=(V,E)$ be a connected graph,
$v\in V$ a vertex, \mbox{$e=\set{v,w}\in E$} an edge, and $F\subseteq E\setminus\set{e}$ a set of edges,
each incident with $v$,
so that $G'\df G-F-e$ is still connected.
Let $B_1$ be those edges that are non-bridges in $G$ but bridges in $G-e$.
Let $B_2$ be those edges that are non-bridges in $G-F$ but bridges in $G-F-e$.
Then $B_1\subseteq B_2$.
\end{proposition}
\begin{proof}
Since $G'$ is connected,
there is a path $(v,e_1,v_1,\hdots,w)$ in $G'$.
Then the cycle $C\df(v,\hdots,w,e,v)$ is in $G-F$.
By \autoref{rem:one-cycle}\ref{rem:one-cycle:ii},
we have $B_1\subseteq E(C)$.
Hence all edges in $B_1$ are on a cycle that is not destroyed by removal of $F$,
so no edge in $B_1$ is made a bridge by removal of $F$.
It follows $B_1\subseteq B_2$.
\end{proof}
\begin{lemma}
\label{lem:simple-minded-convex}%
The simple-minded adversary induces convex cost.
\end{lemma}
\begin{proof}
Let $v\in V$ and $w_1,\hdots,w_k\in V$
and $S$ be a strategy profile.
We show~\eqref{eqn:def-convex-2} \vpageref{eqn:def-convex-2}
proceeding by induction on $k$.
The case $k=1$ is clear.
Let $k>1$ and set $S'\df S-(v,w_1)-\hdots-(v,w_{k-1})$.
We show that switching from $S'$ to $S'-(v,w_k)$
increases disconnection cost for $v$ at least as much as switching from $S$ to $S-(v,w_k)$. 
Since $G(S')$ has fewer edges than $G(S)$,
it suffices to consider changes in relevance~$R(\cdot)$.
\par
When removing $\set{v,w_k}$, relevance of zero or more edges 
changes from $0$ to a positive value;
these are precisely those edges which become bridges by the removal
and which were no bridges before.
No relevance is reduced by removal of edges.
\par
Let $B_1$ be all those edges that become bridges by the switch from $S$ to $S-(v,w_k)$,
and let $B_2$ those that become bridges by the switch from $S'$ to $S'-(v,w_k)$.
Then $B_1\subseteq B_2$ by \autoref{prop:becoming-bridges}.
The increase in relevance from $0$ to a positive value for $e\in B_1$
given $S'$ is at least as high as when given $S$.
In other words,
while $\set{v,w_1},\hdots,\set{v,w_{k-1}}$ are removed,
the effect of all edges in $B_1$ becoming bridges is saved 
until the removal of $\set{v,w_k}$.
We have shown that
\begin{equation}
\label{eqn:convexity-induction-step}%
I_v(S'-(v,w_{k})) - I_v(S')
\geq I_v(S-(v,w_k)) - I_v(S)\enspace.
\end{equation}
The proof is concluded by the following standard calculation:
\begin{align*}
&\alignskip I_v(S-(v,w_1)-\hdots-(v,w_k))-I_v(S)\\
&=I_v(S'-(v,w_k))-I_v(S)\\
&=I_v(S'-(v,w_k))-I_v(S)+I_v(S')-I_v(S')\\
&=I_v(S'-(v,w_k))-I_v(S')+I_v(S')-I_v(S)\\
&\geq I_v(S-(v,w_k))-I_v(S)+I_v(S')-I_v(S) & \text{by~\eqref{eqn:convexity-induction-step}}\\
&\geq I_v(S-(v,w_k))-I_v(S)
+ \sum_{i=1}^{k-1} \big(I_v(S-(v,w_i)) - I_v(S)\big) & \text{by induction} \\
& = \sum_{i=1}^{k} \big(I_v(S-(v,w_i)) - I_v(S)\big)\enspace.
\tag*{\qedhere}
\end{align*}
\end{proof}
\par\medskip
We restrict to the simpler concept of \PS\ in the following study 
of the price of anarchy,
knowing that by convexity it is the same as \PNE.
Recall that we have shown in \autoref{prop:chord-free}
that a \NE\ is chord-free.
The proof does not fully carry over to the bilateral case,
since it contains an argument of the form
\enquote{then the player would rather build a different link instead.}
Yet, we can use the idea of that proof 
to show chord-freeness if $\alpha$ is not too small.
For small $\alpha$, we can show a bound on the number of edges 
by a different simple argument.
\begin{proposition}
\label{prop:chord-free-bilateral-simple}%
Let a pairwise stable graph $G$ be given.
\begin{enumerate}[label=(\roman*)]
\item\label{prop:chord-free-bilateral-simple:i}%
If $\alpha>\frac{1}{2}$, then $G$ is chord-free
and hence by \autoref{prop:chord-free-general} only has $2n-1=O(n)$ edges.
\item\label{prop:chord-free-bilateral-simple:ii}%
In general, $G$ is chord-free
(with $O(n)$ edges)
or has at most $\frac{n}{\sqrt{2\alpha}}+1$ edges.
\end{enumerate}
\end{proposition}
\begin{proof}
If $G$ is bridgeless, selling a chord is beneficial since
disconnection cost $0$ is maintained.
So for both parts we assume that $G$ contains bridges.
\par
\ref{prop:chord-free-bilateral-simple:i}
The impairment in disconnection cost for a player $v$ of selling a chord 
is only due to the change in the denominator of the disconnection cost
and is precisely $\frac{1}{m\,(m-1)}\,R(v)$,
which is upper-bounded by $\frac{1}{2}$
since $R(v)\leq\frac{n\,(n-1)}{2}$.
Hence if $\alpha$ is larger than that,
there is an incentive to sell the chord.
\par
\ref{prop:chord-free-bilateral-simple:ii}
Let $G$ possess a chord.
This means that any of its two endpoints, say $v$, 
deems it being no impairment
to pay $\alpha$ for this edge,
hence $\frac{1}{m\,(m-1)}\,R(v) \geq \alpha$.
It follows 
\begin{equation*}
\frac{n^2}{2} \geq \frac{n\,(n-1)}{2} \geq R(v) 
\geq m\,(m-1)\,\alpha \geq (m-1)^2\,\alpha
\enspace,
\end{equation*}
hence $\frac{n}{\sqrt{2\alpha}}+1\geq m$.
\end{proof}
\par
As for bounding disconnection cost,
\autoref{lem:diam} is no longer true,
but\linebreak \autoref{lem:double-count} and its \autoref{cor:disconnect-diam-bound} is.
If $\al>\frac{1}{2}$,
we can at least show a bound of $O(1+\sqrt{\sfrac{n}{\alpha}})$
on the price of anarchy;
we do not know whether it is tight.
If $\al\leq\frac{1}{2}$,
we can show $O(1+\sqrt{\sfrac{n}{\alpha^{1.5}}})$.
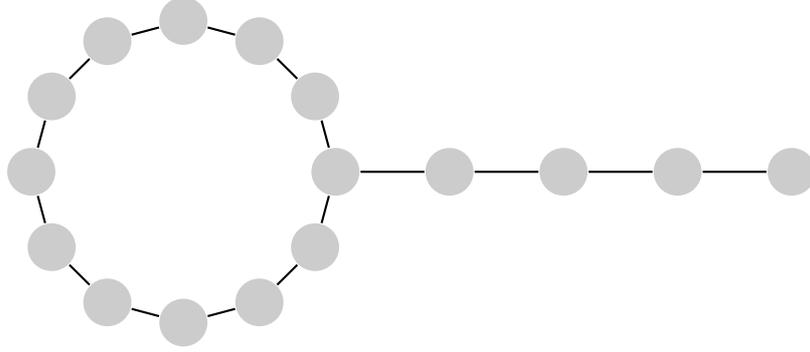
\begin{figure}
\centering
\begin{tikzpicture}[%
	every node/.style=player,%
	node distance=1.5cm,%
	]
	\foreach \i in {0,30,...,330}{
		\node (v\i) at (\i:2cm) { };
	};
	\foreach \i in {0,30,...,300}{
		\pgfmathtruncatemacro{\next}{\i+30}
		\path (v\i) edge (v\next);
	};
	\path (v330) edge (v0);
	\node (r1) [right = of v0] { };
	\node (r2) [right = of r1] { };
	\node (r3) [right = of r2] { };
	\node (r4) [right = of r3] { };
	\path (v0) edge (r1) (r1) edge (r2) (r2) edge (r3) (r3) edge (r4);
\end{tikzpicture}
\caption{%
	\label{fig:bilateral-cycle-path}%
	Cycle with path attached, here $n=16$ and $\ell=4$.}
\end{figure}%
\begin{lemma}
\hfill
\label{lem:bilateral-bridge-tree-diameter}%
\begin{enumerate}[label=(\roman*)]
\item
\label{lem:bilateral-bridge-tree-diameter:upper}%
The diameter of the bridge tree of a \PS\ graph is $O\parens{\sqrt{n\alpha}}$
if $\al>\frac{1}{2}$ 
and $O\parens{\sqrt{n\sqrt{\alpha}}}$ if $\al\leq\frac{1}{2}$.
\item
\label{lem:bilateral-bridge-tree-diameter:lower}%
For $\alpha=1$, the bridge tree can have diameter $\Omega(\sqrt{n})$,
even if the graph is \PS.
\end{enumerate}
\end{lemma}
\begin{proof}
\ref{lem:bilateral-bridge-tree-diameter:upper}~Building an edge that puts a path 
in the bridge tree of length $\ell\geq 1$
on a cycle brings to \emphasis{both endpoints} at least the $\Delta R$ 
of that path alone, \ie $\sum_{k=1}^\ell k = \frac{\ell \, (\ell+1)}{2}$.
If $\al>\frac{1}{2}$, then
the edge so brings an improvement in disconnection cost of at least
$\frac{1}{m+1}\Delta R \geq \frac{1}{2n} \frac{\ell \, (\ell+1)}{2}$,
so $4n\alpha \geq \ell \, (\ell+1)\geq \ell^2$.
We use the $2n-1$ bound on the number of edges from
\autoref{prop:chord-free-bilateral-simple}\ref{prop:chord-free-bilateral-simple:i} here.
If $\al\leq\frac{1}{2}$, then by
\autoref{prop:chord-free-bilateral-simple}\ref{prop:chord-free-bilateral-simple:ii}
we have an $O\parens{\frac{n}{\sqrt{\al}}}$ bound on the number of edges, 
yielding $\ell=O\parens{\sqrt{n\sqrt{\al}}}$ along the same lines.
\par
\ref{lem:bilateral-bridge-tree-diameter:lower}~Consider a cycle
with a path of length $\ell$ attached to it with one of its ends,
as shown in \autoref{fig:bilateral-cycle-path}
\vpageref{fig:bilateral-cycle-path}.
Let $n$ be the total number of vertices
and
let $\frac{1}{n} \frac{\ell \, (\ell+1)}{2}\leq 
\alpha 
\leq \frac{1}{n} \frac{((n-\ell)-1) \, (n-\ell)}{2}$;
such an $\alpha$ exists if $n\geq 3\ell$.
The bridge tree has diameter $\ell$.
Because of the lower bound on $\alpha$,
no vertex on the cycle wishes to connect to a vertex on the path,
and also no vertex on the path wishes to connect to a vertex that is 
located away from the cycle.
Because of the upper bound on $\alpha$,
it can also be shown easily that no two neighboring vertices 
on the cycle wish to sell the edge between them.
Trivially, no edge on the path will be sold.
Hence this graph is \PS.
However, we can choose $n\geq 9$, $\ell\df\sfloor{\sqrt{n}}$, and $\alpha\df 1$,
and so have a diameter of $\Omega(\sqrt{n})$.
\end{proof}
One might hope to find a lower bound matching $O(1+\sqrt{\sfrac{n}{\alpha}})$
by the construction used in this proof.
This is, however, not the case:
\begin{remark}
The example in
\autoref{lem:bilateral-bridge-tree-diameter}%
\ref{lem:bilateral-bridge-tree-diameter:lower}
does not prove
a price of anarchy beyond~$\Theta(1)$.
\end{remark}
\begin{proof}
The total disconnection cost is 
\begin{equation*}
\frac{1}{n}\,2\,\sum_{k=1}^\ell (n-k) \, k
= \frac{\ell\,(\ell+1)}{n} \parens{n-\frac{2\ell+1}{3}}\enspace.
\end{equation*}
For a lower bound on the price of anarchy,
we have to divide this by $\Theta(n\alpha)$,
so we choose $\alpha$ as small as possible, \ie $\alpha=\frac{1}{n}\frac{\ell\,(\ell+1)}{2}$.
Then we divide the total disconnection cost by $\Theta(n\alpha)=\Theta(\ell\,(\ell+1))$ and receive
a lower bound on the price of anarchy of only
\begin{equation*}
\Theta\delim{(}{\frac{1}{n} \parens{n-\frac{2\ell+1}{3}}}{)}
= \Theta\delim{(}{1-\frac{2\ell+1}{3n}}{)}
= O(1)\enspace.
\qedhere
\end{equation*}
\end{proof}
\begin{theorem}
The price of anarchy for a simple-minded adversary with \BLF\
and \PS\
is $O(1+\sqrt{\sfrac{n}{\alpha}})$ if $\al>\frac{1}{2}$
and $O(1+\sqrt{\sfrac{n}{\alpha^{1.5}}})$ otherwise.
\end{theorem}
\begin{proof}
By \autoref{prop:chord-free-bilateral-simple}
building cost of a pairwise stable graph is $O(n\alpha)$
or $O((\frac{n}{\sqrt{2\alpha}} + 1)\,\alpha)$,
both having a ratio of $O(1+\frac{1}{\sqrt{\alpha}})$ to the optimum.
\par
If $\al>\frac{1}{2}$, then
by \autoref{lem:bilateral-bridge-tree-diameter}%
\ref{lem:bilateral-bridge-tree-diameter:upper} and
\autoref{cor:disconnect-diam-bound},
disconnection cost of a pairwise stable graph is $O(n\sqrt{n\alpha})$,
having a ratio of $O(\sqrt{\sfrac{n}{\alpha}})$ to the optimum.
Otherwise, if $\al\leq\frac{1}{2}$,
then disconnection cost is $O\parens{n\sqrt{n\sqrt{\alpha}}}$,
having a ratio of $O\parens{\sqrt{\sfrac{n}{\alpha^{1.5}}}}$ to the optimum.
\end{proof}

\subsection*{Smart Adversary}
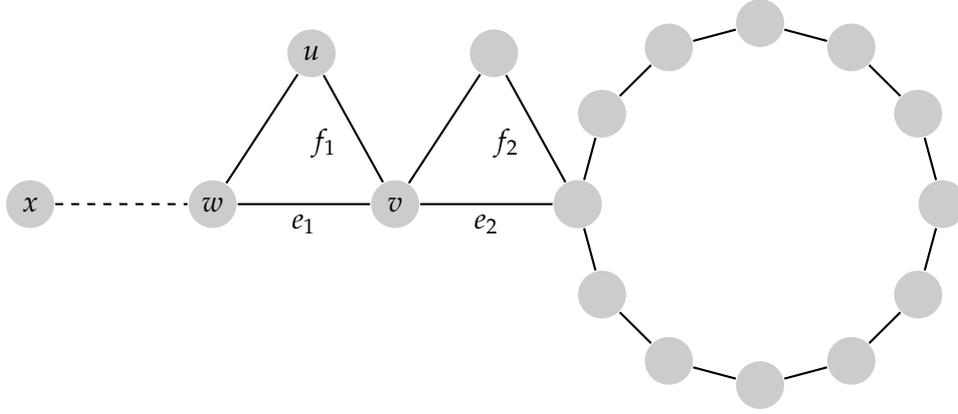
\begin{figure}
\centering
\begin{tikzpicture}
	\tikzset{node distance=2.4cm}
	\node[player] (x) {$x$};
	\node[player] (w) [right = of x] {$w$};
	\node[player] (u) [above right = 2cm and 1.3cm of w] {$u$};
	\node[player] (v) [right = of w] {$v$};
	\node[player] (uu) [above right = 2cm and 1.3cm of v] {};
	\node[player] (y) [right = of v] { };
	\node (c) [right = of y] { };
	\foreach \i in {0,30,...,330}{
		\node[player] (c\i) at ([shift={(\i:2.4cm)}] c) { };
	};
	\foreach \i in {0,30,...,300}{
		\pgfmathtruncatemacro{\next}{\i+30}
		\path (c\i) edge (c\next);
	};
	\path (c0) edge (c330);
	\path
		(x) edge[dashed] (w)
		(w) edge  (u)
		(w) edge node[swap] {$e_1$} (v)
		(u) edge node[swap,pos=.4] {$f_1$} (v)
		(v) edge node[swap] {$e_2$} (y)
		(v) edge (uu)
		(uu) edge node[swap,pos=.4] {$f_2$} (y);
\end{tikzpicture}
\caption{%
	\label{fig:non-convex}%
	Example for a non-convex individual cost $C_v$.
	The dashed edge is critical.
	Since the number of vertices on the cycle on the right is large enough,
	removal of $e_i$ makes $f_i$ critical for each $i\in\set{1,2}$.}
\end{figure}%
We start again by considering convexity of cost.
The difficulty with the smart adversary is 
that its probability measure can change substantially when edges are removed.
Indeed, exploiting this feature we show that cost 
is \emphasis{not} convex for the smart adversary.
\begin{proposition}
Consider the graph in \autoref{fig:non-convex}
and the player $v$ in it.
Then $C_v$ is not convex, since
\begin{equation*}
C_v(G-e_1-e_2)-C_v(G)
< (C_v(G-e_1)-C_v(G)) + (C_v(G-e_2)-C_v(G))\enspace.
\end{equation*}
\end{proposition}
\begin{proof}
Let $k$ be the number of vertices on the cycle on the right.
If $k$ is large enough,
removal of $e_i$ makes $f_i$ critical, for each $i\in\set{1,2}$.
Removing both $e_1$ and $e_2$ makes $f_2$ critical.
Thus we have:
\begin{align*}
C_v(G-e_1) - C_v(G) &= 3 - 1 - \alpha  = 2-\alpha\enspace,\\
C_v(G-e_2) - C_v(G) &= k - 1 - \alpha \enspace\text{,}\enspace\text{and}\\
C_v(G-e_1-e_2) - C_v(G) &= k - 1 - 2\alpha \enspace.
\end{align*}
So $C_v(G-e_1) - C_v(G) + C_v(G-e_2) - C_v(G) = 2 + k - 1 - 2\alpha =k+1-2\alpha$,
which is strictly larger than $k - 1 - 2\alpha$ by a difference of $2$.
\end{proof}
This result is only partly satisfactory,
since non-convexity is not shown on the set of \PS\ strategy profiles.
It remains unclear how to construct an example 
of a \PS\ graph that is not a \PNE.
In an attempt to make the example from \autoref{fig:non-convex} \PS,
we need $\alpha\leq 1$, or else $u$ would sell $\set{u,w}$,
which would make $e_1$ critical.
However, then there is an incentive for $x$ to put the critical edge on a cycle
by building an additional edge,
and no potential partner can decline such a request.
\par
Yet, this example provides evidence that the smart adversary 
is in some respect a substantially different model than the sum-distance model
or the simple-minded adversary.
For, these two have convex cost on the whole set of strategy profiles,
which the smart adversary has not.
\par\medskip
Now we turn to the price of anarchy.
Here again, the smart adversary provides us with something new.
The proof of \autoref{rem:chord-free-smart},
showing that a \NE\ is chord-free,
clearly carries over to the bilateral case and pairwise stability.
This is easier compared to the simple-minded adversary,
where we had to prove \autoref{prop:chord-free-bilateral-simple}.
Thus certainly we have the 
\begin{equation}
\label{eqn:simple-bound}%
O\parens{1+\frac{n}{\alpha}}
\end{equation}
bound of \autoref{cor:simple-bound}\ref{cor:simple-bound:ii},
since chord-freeness is maintained.
Can it be improved, like in all the other cases studied before?
The proofs of \autoref{lem:mmax-3} and \autoref{lem:mmax-2-1}, 
which are the basis for \autoref{thm:smart-bound},
are almost completely symmetric in $v_1$ and $v_2$
and so at first appear to apply for \BLF\ as well,
which would have meant an $O(1)$ bound.
The case of exactly one critical edge in \autoref{lem:mmax-2-1}, however, 
is not symmetric and provides an idea to a counterexample:
if there is only one critical edge~$e_0$,
then it can happen that vertices in the smaller component of $G-e_0$
wish (desperately) to put~$e_0$ on a cycle,
but they cannot find a partner in the other component 
that is willing to cooperate.
The following lower bound is tight by~\eqref{eqn:simple-bound}.
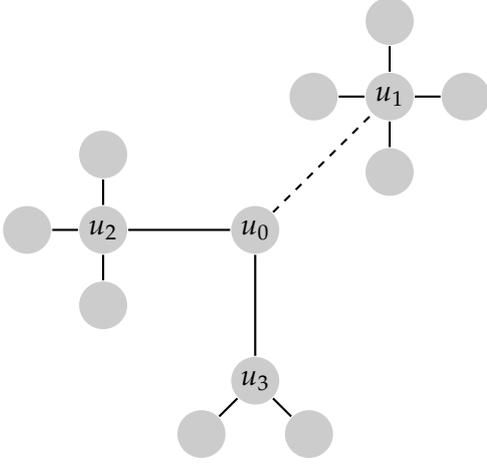
\begin{SCfigure}
\caption{%
\label{fig:three-stars}%
Three stars of sizes $n_0$, $n_0-1$, and $n_0-2$; here $n_0=5$.
The $n_0$ players in the star around $u_1$ would like to put the one critical edge on a cycle,
if $\alpha<n_0$.
Building, \eg $\set{u_1,u_2}$ would reduce their disconnection cost from $n-n_0$ to $n_0-2$,
meaning an improvement of $n_0$.
But no player from the stars around $u_2$ or $u_3$ is willing to cooperate.}
\begin{tikzpicture}[%
	every node/.style=player,%
	node distance=1cm,%
	]
	\node (u0) {$u_0$};
	\node (u1) [above right = 2.5cm of u0] {$u_1$};
	\node (u11) [below = of u1] { };
	\node (u12) [right = of u1] { };
	\node (u13) [above = of u1] { };
	\node (u14) [left = of u1] { };
	\node (u2) [left = 2cm of u0] {$u_2$};
	\node (u21) [below = of u2] { };
	\node (u22) [left = of u2] { };
	\node (u23) [above = of u2] { };
	\node (u3) [below = 2cm of u0] {$u_3$};
	\node (u31) [below right = of u3] { };
	\node (u32) [below left = of u3] { };
	\path (u0) edge[dashed] (u1)
		  (u0) edge         (u2)
		  (u0) edge         (u3);
	\path (u1) edge (u11)
		  (u1) edge (u12)
		  (u1) edge (u13)
		  (u1) edge (u14);
	\path (u2) edge (u21)
		  (u2) edge (u22)
		  (u2) edge (u23);
	\path (u3) edge (u31)
		  (u3) edge (u32);
\end{tikzpicture}
\hspace{1cm}
\end{SCfigure}%
\begin{theorem}
\label{thm:bilateral-smart-lower}%
Let $n \geq 9$ and $\al > 2$, 
then the price of anarchy for \PNE\ and \PS\ with the smart adversary 
is $\Omega(1 + \frac{n}{\alpha})$.
\end{theorem}
\begin{proof}
First assume there is an integer $n_0\geq 3$ such that $n=3n_0-2$.
Consider three stars $S_i$, $i=1,2,3$ with center vertices $u_i$, $i=1,2,3$,
and $n_0$, $n_0-1$, and $n_0-2$ vertices, respectively.
Connect the stars via an additional vertex $u_0$ 
and additional edges $\set{u_0,u_i}$, $i=1,2,3$.
This construction uses $3n_0-2$ vertices.
See \autoref{fig:three-stars} 
for an illustration.
Then $e_0=\set{u_0,u_1}$ is the only critical edge
and $n_0=\Theta(n)$,
namely slightly more than $\frac{1}{3}n$.
We have a total disconnection cost of $2 \nn(e_0) \, (n-\nn(e_0))
= 2 n_0 \, (n-n_0) = \Omega(n^2)$.
We have a ratio to the optimum of $\Omega(1+\frac{n}{\alpha})$.
We are left to show that
this graph is a \PNE,
which implies \PS.
It is clear that no edge can be sold,
since that would make the graph disconnected.
Therefore we only have to ensure 
that no link can be added that would be beneficial for one endpoint
and at least no impairment for the other one,
\ie we have to show~\eqref{eqn:pne}.
\par
An edge $e$ that improves disconnection cost for some player 
has to put $\set{u_0,u_1}$ on a cycle.
If $e$ connects a vertex in $S_1$ with a vertex in $S_2$,
then $\set{u_0,u_3}$ will become critical.
For a vertex in $S_2$,
this reduces disconnection cost from $n_0$ to $n_0-2$.
So, since $\alpha> 2$, no vertex in $S_2$ agrees to build such an edge.
\par
A similar situation holds if $e$ connects a vertex in $S_1$ with a vertex in $S_3+u_0$.
It will result in $\set{u_0,u_2}$ becoming critical.
For a vertex in $S_3+u_0$,
this reduces disconnection cost from $n_0$ to $n_0-1$.
So, since $\alpha> 1$, no vertex in $S_3+u_0$ agrees to build such an edge.
\par\smallskip
If $n+2$ is not a multiple of $3$, we can do a similar construction.
We let $n_0$ be as large as possible so that $n\geq3n_0-2$
and do the same construction as above.
The remaining $1$ or $2$ vertices are connected directly to $u_0$.
Then the previous arguments essentially carry over.
\end{proof}
\par
\enlargethispage{\baselineskip}%
If we consider $\alpha>2$ a constant,
then the previous theorem gives a lower bound of $\Omega(n)$.
For $\alpha=\Omega(1)$ this
is the worst that can happen for \emphasis{any} adversary in this model.
Hence the `overall worst-case' is attained by the smart adversary
with \BLF.

\section*{Bibliographic Information}
This work is based on Chapter 5 of my dissertation~\cite{Kli10a}.
It has also been presented in form of a brief announcement at PODC 2010~\cite{Kli10b}.
The results on unilateral link formation have been published
in the open access journal \textit{Games}~\cite{Kli11a}.

\section*{Acknowledgements}
I thank the German Research Foundation
(Deutsche Forschungsgemeinschaft\linebreak (DFG))
for financial support through 
Priority Program 1307 \enquote{Algorithm Engineering} (Grant SR7/12-2).
I thank Ines Kusemeier for pointing out some errors in an early version of this work.
I thank Peter Munstermann for improving the constant in \autoref{prop:chord-free-general},
resulting in smaller constants in the bounds on the price of anarchy.

\preparebibliography
\newcommand{\etalchar}[1]{$^{#1}$}

\end{document}